\documentclass[screen,acmsmall,nonacm]{acmart}

\AtBeginDocument{%
  }

\setcopyright{acmcopyright}
\copyrightyear{2018}
\acmYear{2018}
\acmDOI{XXXXXXX.XXXXXXX}

\acmJournal{TEAC}
\acmVolume{37}
\acmNumber{4}
\acmArticle{111}
\acmMonth{8}

\usepackage{cleveref}
\usepackage{xspace}
\usepackage{tikz}
\usetikzlibrary{calc,backgrounds}
\usepackage{boxedminipage}
\usepackage{todonotes}
\usepackage{tabularx}
\usepackage{booktabs}
\usepackage{multirow}

\newcolumntype{Y}{>{\centering\arraybackslash}X}
\usepackage[ruled]{algorithm} %

\usepackage{algpseudocodex}

\newcommand{\N}{\mathbb{N}}

\newcommand{\agents}{\ensuremath{\mathcal{N}}}
\newcommand{\RF}{\ensuremath{R}}     %
\newcommand{\IH}{\ensuremath{I}}     %
\newcommand{\ia}{\ensuremath{\iota}} %
\newcommand{\pl}{\ensuremath{\pi}}   %

\newcommand{\probName}[1]{\textsc{#1}\xspace}

\newcommand{\ARH}{\probName{Anonymous Refugees Housing}}
\newcommand{\ARHshort}{\textsf{ARH}\xspace}
\newcommand{\HRH}{\probName{Hedonic Refugees Housing}}
\newcommand{\HRHshort}{\textsf{HRH}\xspace}
\newcommand{\DRH}{\probName{Diversity Refugees Housing}}
\newcommand{\DRHshort}{\textsf{DRH}\xspace}

\newcommand{\NP}{\textsf{NP}\xspace}
\newcommand{\paraNPh}{\textsf{para-\NP}-hard\xspace}
\newcommand{\NPh}{\NP-hard\xspace}
\newcommand{\NPhness}{\NP-hardness\xspace}
\newcommand{\NPc}{\NP-complete\xspace}
\newcommand{\W}[1][1]{\textsf{W[#1]}\xspace}
\newcommand{\Wh}[1][1]{\W[#1]-hard\xspace}
\newcommand{\Wc}[1][1]{\W[#1]-complete\xspace}
\newcommand{\XP}{\textsf{XP}\xspace}
\newcommand{\FPT}{\textsf{FPT}\xspace}

\newcommand{\oh}[1]{{o(#1)}}
\newcommand{\Oh}[1]{{\mathcal{O}(#1)}}

\newcommand{\yes}{\emph{yes}\xspace}
\newcommand{\yesI}{\yes-instance\xspace}
\newcommand{\no}{\emph{no}\xspace}
\newcommand{\noI}{\no-instance\xspace}

\newtheorem{remark}{Remark}

\usepackage{pifont}%
\newcommand{\cmark}{\ding{51}}%
\newcommand{\xmark}{\ding{55}}%

\begin{document}

\title{Host Community Respecting Refugee Housing}
\thanks{An extended abstract of this work has been published in the Proceedings of the 22nd International Conference on Autonomous Agents and Multiagent Systems (AAMAS~'23)~\cite{KnopS2023a}.}

\author{Dušan Knop}
\email{dusan.knop@fit.cvut.cz}
\orcid{0000-0003-2588-5709}
\author{Šimon Schierreich}
\email{simon.schierreich@fit.cvut.cz}
\orcid{0000-0001-8901-1942}
\affiliation{%
	\department{Department of Theoretical Computer Science, Faculty of Information Technology}
	\institution{Czech Technical University in Prague}
	\streetaddress{Thákurova 9}
	\city{Prague}
	\country{Czechia}
	\postcode{160 00}%
}

\begin{abstract}
  We propose a novel model for refugee housing respecting the preferences of the accepting community and refugees themselves. In particular, we are given a topology representing the local community, a set of inhabitants occupying some vertices of the topology, and a set of refugees that should be housed on the empty vertices of the graph. Both the inhabitants and the refugees have preferences over the structure of their neighborhood.
  
  We are specifically interested in the problem of finding housing such that the preferences of every individual are met; using {game-theoretical} words, we are looking for housing that is stable with respect to some well-defined notion of stability. We investigate conditions under which the existence of equilibria is guaranteed and study the computational complexity of finding such a stable outcome. As the problem is \NPh even in very simple settings, we employ the parameterized complexity framework to give a finer-grained view of the problem's complexity with respect to natural parameters and structural restrictions of the given topology.
\end{abstract}




\received{20 February 2007}
\received[revised]{12 March 2009}
\received[accepted]{5 June 2009}

\maketitle

\section{Introduction} 

According to the 2023 report of the United Nations High Commissioner for Refugees (UNHCR), there were 108.4 million forcibly displaced persons at the end of 2022~\cite{UNHCR2022}. It is the highest number since the aftermath of World War II, and it is certain that these numbers will grow even more. %
They identified the war in Ukraine as the leading cause of the dramatic growth in the last year~\cite{UN2022}.  Russian aggression not only forced many Ukrainians to leave their homes, but even caused food insecurity and related population movement in many parts of the world, since Ukraine is among the fifth largest wheat exporters in the world~\cite{BehnassiH2022}.  

It should be mentioned that political and armed conflicts are not the only causes of forced displacement~\cite{UNHCR2022}. One of the most common reasons for fleeing is due to natural disasters. To name just a few, in August 2022, massive floods across Pakistan affected at least two-thirds of the districts and displaced at least 33 million people~\cite{Mallapaty2022a,Mallapaty2022b}. At the same time, a devastating drought in Somalia caused the internal displacement of at least 755,000 people~\cite{UNHCR2022}. Furthermore, it is expected that, due to climate change, extremes of climate will become even more common in the near future~\cite{FischerSK2021}.

Arguably, the best prevention against the phenomenon of forced displacement is not allowing it to appear at all; however, the aforementioned numbers clearly show that these efforts are not very successful. Therefore, in practice, three main solutions are assumed~\cite{JallowM2004}. Voluntary \emph{repatriation} is the most desirable but not very successful option. In many situations, repatriation is not even possible due to ongoing conflicts or a completely devastated environment. \emph{Resettlement} and \emph{integration} in the country of origin or abroad are more common. These two solutions require considerable effort from both the newcomers and the host community sides.

The very problematic part of forced displacement is the fact that~$38\%$ of all refugees\footnote{From the strict sociological point-of-view, not all forcibly displaced persons are classified as refugees. Slightly abusing the terminology, we will use the terms refugee and displaced person interchangeably.} are hosted in only five countries~\cite{UNHCR2022}. And these are only the absolute numbers. For example, in Lebanon, every one in four people is a refugee~\cite{NRC2022}. The redistribution of refugees seems to be a natural solution to this imbalance; however, not all countries are willing to accept all people. One such example is the Czech Republic, which refused to accept any Syrian refugees during the 2015 European migrant crisis, hosting the largest number of Ukrainian refugees per capita as of 2023~\cite{UNHCR2022b}.

Even with working and widely accepted redistributing policies, there is still a need to provide housing in specific cities and communities. From the good examples of such integration strategies~\cite{OECD2018,ZierschDW2023}, it follows that one of the most important characteristics is that members of the accepting community do not feel threatened by the newcomers.

Inspired by this, we propose a novel computational model for refugee housing. Our ultimate goal is to find an assignment of displaced persons into empty houses of a community such that this assignment corresponds to the preferences of the inhabitants about the structure of their neighborhoods and, at the same time, our model also takes into consideration the preferences of the refugees themselves, as refugees dissatisfied with their neighborhood have a strong intention to leave the community. More precisely, in our model, we are given a topology of the community, which is an undirected graph, a set of inhabitants together with their assignment to the vertices of the topology and preferences over the shape of their neighborhood, and a set of refugees with the same requirements on the neighborhood shape. We want to find housing for refugees in the empty vertices of the topology so that the housing satisfies a certain criterion, such as stability.

Refugee redistribution has gained the attention of mathematicians and computer scientists only very recently. The formal model for capturing refugee resettlement is a double-sided matching \cite{DelacretazKT2023,AzizCGS2018}. That is, in the input, we are given a set of locations with multidimensional constraints and refugees with multidimensional features. An example of a constraint can be the number of refugees the location can accept on the one side and the size of a family on the refugee side. The question then is whether there exists a matching between locations and refugees respecting all constraints. According to us, this formulation of the refugee resettlement problem concerns the global perspective of refugee redistributing, not the local housing problem, as we focus on in our paper. \citet{AzizCGS2018} study mostly the complexity of finding stable matching with respect to different notions of stability; it turns out that, for most of the stability notions, finding a stable matching is computationally intractable (\NPh, in fact). \citet{KuckuckRW2019} later refined the model of \citet{AzizCGS2018} in terms of hedonic games. \citet{BansakFHDHLW2018,BansakP2022,AhaniAMTT2021} explored the application of machine-learning techniques in the context of refugee redistribution.

\subsection{Our Contribution}

Partly continuing the line of research in refugee resettlement, we introduce a novel model focused on the local housing of new refugees. Previous models~\cite{AzizCGS2018,AhaniGPTT2023,AhaniAMTT2021,DelacretazKT2023} can be seen and used as a very effective model on the (inter-) national level to distribute refugees to certain locations, such as states or cities.\footnote{In fact, the American resettlement agency HIAS uses the matching software \emph{Annie\textsuperscript{TM}} \textsc{Moore} which is powered by the refugee-redistribution algorithms and ILP formulations of the problem~\cite{AhaniAMTT2021,AhaniGPTT2023}.} However, our model can be assumed as the second level of refugee redistribution; once refugees are allocated to some community, we want to house them in a way that respects the preferences of both inhabitants and refugees.

\begin{table}[bt!]
	\caption{A basic overview of our complexity results. The first column contains specific variants of the refugee housing problem (\textsf{ARH} -- anonymous preferences, \textsf{HRH} -- hedonic preferences, and \textsf{DRH} -- diversity preferences). All the other columns contain complexity classification of the combination of a parameter and a variant; here,~$\RF$ is a set of refugees,~$\IH$ is a set of inhabitants,~$\operatorname{vc}$ is the vertex-cover number, and~$\operatorname{tw}$ is the tree-width of the topology, respectively. If a cell contains \NPh for some parameter, it means that the problem is \NPh already for a constant value of this parameter.}
	\label{fig:results}
	\Crefname{theorem}{Thm.}{Thm.}
	\Crefname{corollary}{Cor.}{Cor.}
	\renewcommand{\crefpairconjunction}{,\,}
	\def\arraystretch{1.2}
	\begin{center}
		\begin{tabularx}{\textwidth}{lYYYYYY}
			\toprule
			& 
				--- & 
				$|\RF|$ & 
				$|\IH|$ & 
				$|\RF|+|\IH|$ &
				~$\operatorname{vc}$ &
				~$\operatorname{tw}$ \\
			\midrule
			\ARHshort & 
				\NP-h {\small\Cref{thm:ARH:NPh}} & 
				\mbox{\W[2]-h\,+\,\XP} {\small\Cref{thm:ARH:Wh:R,thm:ARH:XP:R}} & 
				{\W-h\,+\,\XP} {\small\Cref{thm:ARH:Wh:I,thm:ARH:XP:I}} & 
				\FPT{} {\small\Cref{thm:ARH:FPT:R+I}} & 
				{\W-h\,+\,\XP} {\small\Cref{thm:ARH:Wh:vc,thm:ARH:XP:vc}} & 
				\W-h {\small\Cref{thm:ARH:Wh:vc}}\\
			\midrule
			\HRHshort & 
				\NP-h {\small\Cref{thm:HRH:NPh}} & 
				{\W[2]-h\,+\,\XP} {\small\Cref{thm:HRH:XP:R}} & 
				{\NP-h} {\small\Cref{thm:HRH:NPh:I}} & 
				\FPT {\linebreak\small\Cref{thm:HRH:FPT:R+I}} & 
				{\NP-h} {\small\Cref{thm:HRH:NPh:vc}} & 
				\NP-h {\small\Cref{thm:HRH:NPh:vc}} \\
			\midrule
			\DRHshort & 
				\NP-h {\small\Cref{thm:DRH:NPc:types}} & 
				{\W[2]-h\,+\,\XP} {\Cref{thm:DRH:XP:R}} &
				{\NP-h} {\small\Cref{thm:DRH:HRH:reduction}}&
				{\W[1]-h\,+\,\XP} {\Cref{thm:DRH:Wh:IR,thm:DRH:XP:R}} &
				{\NP-h} {\small\Cref{thm:DRH:HRH:reduction}}&
				{\NP-h} {\small\Cref{thm:DRH:HRH:reduction}}\\
			\bottomrule
		\end{tabularx}
	\end{center}
	\renewcommand{\crefpairconjunction}{and \nobreakspace}
\end{table}

In particular, we introduce three variants of refugee housing, each targeting a certain perspective of this problem. Our simplest model, introduced in \Cref{sec:anonymous_refugees}, completely eliminates the preferences of refugees and studies only the stability of the housing with respect to the preferences of the inhabitants. We call this variant \emph{anonymous} housing. Since refugees are assumed to be indistinguishable, inhabitants have preferences over the number of refugees in their neighborhood.

As stated above, the most successful refugee integration projects have the following properties in common: they try to make both inhabitants and refugees as satisfied as possible through various activities to ensure that both groups get to know each other. We believe that our \emph{hedonic} model, where the preferences of both inhabitants and refugees are based on the identity of particular members of the other group, supports and leads to more stable and acceptable housing. This model is formally defined and studied in \Cref{sec:hedonic_setting}.

The two introduced models have some disadvantages. The first is disrespectful to the refugees' preferences, while the second is not very realistic, as it is hard to make all inhabitants familiar with all refugees and the other way around. Therefore, our last model can be seen as a compromise between these two extremes. In the \emph{diversity} setting, introduced in \Cref{sec:diversity_setting}, all agents (both inhabitants and refugees) are partitioned into~$k$ types, and their preferences are over the fractions of agents of each type in the neighborhood of each agent. Another advantage of this approach is that it nicely captures also the settings where we already have some number of integrated refugees and the newcomers want to have some of them in the neighborhood, or the case of an internally displaced person, where naturally some inhabitants and refugees share some similarities.

In all the aforementioned variants of the refugee housing problem, agents have dichotomous preferences; that is, they approve some set of alternatives and do not distinguish between them. It can be seen that if the neighborhood of some agent does not comply with his approval set, he would rather leave the local community, which is very undesirable behavior on both sides. 

For all assumed variants, we show that an equilibrium is not guaranteed to exist even in very simple instances. Thus, we study the computational complexity of finding an equilibrium or deciding that no equilibrium exists. To this end, we provide polynomial-time algorithms and complementary \NPhness results. In order to paint a more comprehensive picture of the computational tractability of the aforementioned problems, we employ a finer-grained framework of parameterized complexity to give tractable algorithms for, e.\,g., instances where the number of refugees or the number of inhabitants is small or for certain structural restrictions of the topology. Additionally, we complement many of our algorithmic results with conditional lower bounds matching the running time of these algorithms. A basic overview of our results can be found in \Cref{fig:results}.

\subsection{Related Work} 
Our model is influenced by a game-theoretic reformulation of the famous Schelling's model~\cite{Schelling1969,Schelling1971} of residential segregation introduced by \citet{AgarwalEGISV2021}. Here, we are given a simple undirected graph~$G$ and a set of selfish agents partitioned into~$k$ types. Every agent wants to maximize the fraction of agents of her own type in her neighborhood. The goal is then to assign agents to the vertices of~$G$ so that no agent can improve her utility by either jumping to an unoccupied vertex or swapping positions with another agent. Follow-up works include those that study the problem from the perspective of computational complexity and equilibrium existence guarantees~\cite{KreiselBFN2022,EchzellFLMPSSS2019,DeligkasEG2023,FriedrichLMS2023,BiloBDLMS2023,BiloBLM2022a,BiloBLM2022b,KanellopoulosKV2023}.

The second main inspiration for our model is the \probName{Hedonic Seat Arrangement} problem and its variants recently introduced by \citet{BodlaenderHJOOZ2020}. Here, the goal is to find an assignment of agents with preferences for the vertices of the underlying topology. The desired assignment should then meet specific criteria such as different forms of stability, maximizing social welfare, or being envy-free. In our model, compared to \probName{Hedonic Seat Arrangement} of \citeauthor{BodlaenderHJOOZ2020}~\cite{BodlaenderHJOOZ2020}, the inhabitants already occupy some vertices of the topology, and we have to assign refugees to the remaining (empty) vertices in a desirable way. \probName{Hedonic Seat Arrangement} is also heavily studied from the perspective of (parameterized) algorithms and complexity~\cite{CeylanCR2023,Wilczynski2023,BerriaudCW2023}.

Next, the problem of \emph{house allocation}~\cite{AbdulkadirogluS1999} or \emph{housing market}~\cite{ShapleyS1974} has been extensively studied in the area of mechanism design. Here, each agent owns a house, and the objective is to find a socially efficient outcome using reallocations of objects. Later, \mbox{\citet{YouDTLY2022}} introduced house allocation over social networks that follows the current trend in mechanism design initiated by \citet{LiHZZ2017}. There, each individual can only communicate with his neighbors. As stated before, house allocation is studied mainly from the viewpoint of mechanism design, and as such, it is far from our model.

Then, \emph{hedonic games}~\cite{DrezeG1980,BogomolnaiaJ2002,BrandtCELP2016} are a well-studied class of coalition formation games where the goal is to partition agents into coalitions and where the utility of every agent depends on the identity of other agents in his coalition. 
In \emph{anonymous games}~\cite{BanerjeeKS2001,BogomolnaiaJ2002}, the agents have preferences over the sizes of their coalition. The most recent variants of hedonic games are the so-called \emph{hedonic diversity games}~\cite{BredereckEI2019,BoehmerE2020b,Darmann2023,GanianHKSS2023} where agents are partitioned into~$k$ types and preferences are over the ratios of each type in the coalition. The main difference between our model and (all variants of) hedonic games is that in the latter model, all coalitions are pairwise disjoint; however, in our case, each agent has his own neighborhood overlapping with neighborhoods of other agents.
In closely related \emph{social distance games}~\cite{BranzeiL2011,BalliuFMO2017,BalliuFMO2019,BalliuFMO2022,GanianHKRSS2023}, there is also the topology. However, the position of agents in the topology is fixed, and the goal is to partition the agents into stable coalitions, similar to the model of hedonic games.

Finally, in the recently introduced \emph{topological distance games}~\cite{BullingerS2024,DeligkasEKS2024}, the input also consists of a topology and a set of agents. The goal is to assign the agents to the topology so that no agent wants to perform some deviation, such as swapping positions with other agents or jumping to an empty vertex. However, in topological distance games, the utility of every agent is based on the inherent utility the agent has for other agents and their distance in the assignment; in our setting, agents are interested only in their neighborhood.

\section{Preliminaries}\label{sec:preliminaries}

Let~$\N$ denote the set of positive integers. Given two positive integers~${i,j\in\N}$, with~$i \leq j$, we call the set~$[i,j] = \{i,\ldots,j\}$ an \emph{interval}, and we let~$[i] = [1,i]$ and~$[i]_0=[i]\cup\{0\}$. Let~$S$ be a set. By~$2^S$ we denote the set of all subsets of~$S$ and, given~$k\leq |S|$, we denote by~$\binom{S}{k}$ the set of all subsets of~$S$ of size~$k$.

\subsection{Graph Theory}
All graphs assumed in this work are simple and undirected. Formally, a graph~$G$ is a pair~$(V,E)$, where~$V$ is a non-empty set of \emph{vertices} and~$E\subseteq\binom{V}{2}$ is a set of \emph{edges}.
Given a vertex~$v\in V$, we denote by~$N_G(v)$ the set of its \emph{neighbors}, formally,~$N_G(v) = \{u\mid \{u,v\}\in E\}$. The size of the neighborhood of a vertex~$v$ is called its \emph{degree} and is defined as~$\deg(v) = |N_G(v)|$. The \emph{closed neighborhood} of vertex~$v$ is defined as~$N_G[v] = N_G(v) \cup \{v\}$. 
In this work, we follow the basic graph-theoretical terminology given in~\cite{Diestel2017}.

\subsection{Refugee Housing}
Let~$\RF = \{r_1,\ldots,r_m\}$ be a non-empty set of \emph{refugees} and~$\IH = \{h_1,\ldots,h_\ell\}$ be a set of \emph{inhabitants}. The set of all \emph{agents} is defined as~$\agents = \RF \cup \IH$. We set~$n=|V|$.
A \emph{topology} is a simple undirected graph~$G=(V,E)$, where~$|V| \geq |\agents|$. 
An \emph{inhabitants assignment} is an injective function~$\ia\colon \IH\to V$ that maps inhabitants to vertices of the topology. The set of vertices occupied by the inhabitants is denoted~$V_\IH$ and, given an inhabitant~$h\in \IH$, we denote the set of unoccupied vertices in his neighborhood~$U_h = N(\ia(h))\setminus V_\IH$. The set of all vertices that are not occupied by inhabitants is denoted~$V_U = V\setminus V_I$.
The goal of every variant of our problem is to find a mapping of refugees to vertices that are not occupied by inhabitants. Formally, \emph{housing} is an injective mapping~$\pl\colon \RF \to V_U$. A set of vertices occupied by refugees with respect to housing~$\pl$ is denoted~${V_\pl = \{\pl(r)\mid r\in \RF\}}$. We denote by~$\Pi_{G,\ia}$ the set of all possible housings, and we drop the subscript whenever~$G$ and~$\ia$ are clear from the context.%

\subsection{Parameterized Complexity}
We study the problem in the framework of parameterized complexity~\cite{CyganFKLMPPS2015,DowneyF1995,Niedermeier2006}. Here, we investigate the complexity of the problem not only with respect to an input size~$n$ but even assuming some additional \emph{parameter}~$k$. The goal is to find a parameter that is small, and the ``hardness'' can be confined to this parameter. The most favorable outcome is an algorithm with running time~$f(k)\cdot n^{\Oh{1}}$, where~$f$ is any computable function. We call this algorithm \emph{fixed-parameter tractable}, and the complexity class containing all problems that admit algorithms with such running time is called \FPT. 
Not all combinations of parameters yield to fixed-parameter tractable algorithms. A less favorable outcome is an algorithm running in~$n^{f(k)}$ time, where~$f$ is any computable function. Parameterized problems admitting such algorithms belong to complexity class~\XP. To exclude the existence of a fixed-parameter tractable algorithm, one can show that the parameterized problem is \Wh[t] for some~$t\geq 1$. This can be done via a \emph{parameterized reduction} from any problem known to be \Wh[t].

\begin{definition}[Parameterized reduction~\cite{CyganFKLMPPS2015}]
	Let~$P$ and~$Q$ be two parameterized problems. A \emph{parameterized reduction} of the problem~$P$ to the problem~$Q$ is an algorithm~$\mathcal{A}$ that, given an instance~$(x,k)$ of~$P$, constructs an instance~$(x',k')$ of~$Q$ such that
	\begin{enumerate}
		\item~$(x,k)$ is a \yesI of~$P$ if and only if~$(x',k')$ is a \yesI of~$Q$,
		\item~$k' \leq g(k)$ for some computable function~$g$, and
		\item~$\mathcal{A}$ runs in \FPT time with respect to~$k$.
	\end{enumerate}
\end{definition}

It could also be the case that a parameterized problem is \NPh even for a fixed value of~$k$; we call such problems \paraNPh and, assuming~$\mathsf{P}\not=\NP$, such problems do not admit \XP algorithms.

\paragraph{Exponential-Time Hypothesis}
Our running-time lower bounds are based on the well-known Exponential-Time Hypothesis (ETH) of \citet{ImpagliazzoP2001}; see also \citet{ImpagliazzoPZ2001} and the survey of \citet{LokshtanovMS2011}. This conjecture states that roughly speaking, there is no algorithm solving \textsf{3-SAT} in time sub-exponential in the number of variables. Our results in this direction rely on the following theorems.

\begin{theorem}[\cite{ChenCFHJKX2005}]\label{thm:ETH:XP}
	Unless ETH fails, none of the following problems admits an algorithm running in time~$f(k)\cdot n^{o(k)}$ for any computable function~$f$: \probName{Dominating Set}, \probName{Multicolored Clique}, \probName{Grid Tiling}, and \probName{Set Cover}.
\end{theorem}

For fixed-parameter algorithms, we show running-time lower bounds using the following result.

\begin{theorem}[\cite{ImpagliazzoP2001,ImpagliazzoPZ2001}]\label{thm:ETH:SAT}
	Unless ETH fails, the \probName{$2$-Balanced~$3$-SAT} problem cannot be solved in~$2^{o(\xi+\mu)}$ time, where~$\xi$ is the number of variables and~$\mu$ is the number of clauses of the input formula~$\varphi$, respectively.
\end{theorem}

\subsection{Structural Parameters}
Let~$G=(V,E)$ be a graph. A set~$C\subseteq V$ is \emph{vertex cover} of~$G$ if~$G\setminus C$ is an edgeless graph. The \emph{vertex cover number}~$\operatorname{vc}(G)$ is the minimum size vertex cover in~$G$.

\begin{definition}[Tree decomposition]\label{def:treeDecomposition}
	A \emph{tree decomposition} of a graph~$G=(V,E)$ is a triple~$\mathcal{T} = (T, \beta, r)$, where~$T$ is a tree rooted at node~$r$ and~$\beta \colon V(T) \to 2^{V}$ is a mapping that satisfies:
	\begin{enumerate}
		\item~$\bigcup_{x \in V(T)} \beta(x) = V$;
		\item For every~$\{u, v\} \in E$ there exists a node~$x \in V(T)$, such that~$u, v \in \beta(x)$;
		\item For every~$u \in V$ the nodes~$\{x \in V(T) \mid u \in \beta(x)\}$ form a connected sub-tree of~$T$.
	\end{enumerate}
\end{definition}

To distinguish between the vertices of a tree decomposition and the vertices of the underlying graph, we use the term \emph{node} for the vertices of a given tree decomposition.

The \emph{width} of a tree decomposition~$\mathcal{T}$ is~${\max_{x \in V(T)} |\beta(x)|-1}$.
The \emph{tree-width} of a graph~$G$, denoted~$\operatorname{tw}(G)$, is the minimum width of a~tree decomposition of~$G$ over all possible tree-decompositions of the graph~$G$.

\section{Anonymous Preferences}\label{sec:anonymous_refugees}%

In our simplest model of refugee housing, we assume refugees are non-strategic, and we are concerned only with the preferences of inhabitants. In this sense, the refugees are, from the viewpoint of inhabitants, anonymous, and the preferences only take into account the number of refugees in the neighborhood of each inhabitant. Similar preferences have already been studied in different problems, such as anonymous hedonic games~\cite{BanerjeeKS2001,BogomolnaiaJ2002,Ballester2004}.

We formally capture this setting in the computational problem called the \ARH problem (\ARHshort for short). A preference of every inhabitant~$h\in\IH$ is a non-empty set~$A_h\subseteq [\deg(\ia(h))]_0$ of the approved numbers of refugees in the neighborhood. Our goal is to decide whether there is a housing~$\pl\colon R \to V_U$ that respects the preferences of all inhabitants.

\begin{definition}
	A housing~$\pl\colon\RF\to V_U$ is called \emph{inhabitant-respecting} if for every~$h\in\IH$ we have~$|N_G(\ia(h))\cap V_\pl| \in A_h$.
\end{definition}

If the approval set~$A_h$ for an inhabitant~$h\in\IH$ consists of consecutive numbers, we say that the inhabitant~$h$ approves an interval.
Also, as the refugees are indistinguishable, we will sometimes use~$\pl$ as a set of empty vertices of size~$|\RF|$ instead of a mapping.

\begin{example}\label{ex:ARH}
    \begin{figure}
        \centering
        \begin{tikzpicture}
            \node[draw,white,fill=black,circle,minimum width=0.7cm,label=180:{\normalsize$\{0,1\}$}] (v1) at (0,0) {$h_1$};
            \node[draw,white,fill=black,circle,minimum width=0.7cm,label=180:{\normalsize$\{0\}$}] (v2) at (0,1.5) {$h_2$};
            \node[draw,circle,minimum width=0.7cm] (v3) at (1.5,1.5) {$r$};
            \node[draw,circle,minimum width=0.7cm] (v4) at (1.5,0) {};

            \node at (0.75,-1) {\textcolor{red!80!black}{\Large\xmark}};
            
            \draw (v1) -- (v2) -- (v3) -- (v4) -- (v1);
        \end{tikzpicture}
        \hspace{1cm}
        \begin{tikzpicture}
            \node[draw,white,fill=black,circle,minimum width=0.7cm,label=180:{\normalsize$\{0,1\}$}] (v1) at (0,0) {$h_1$};
            \node[draw,white,fill=black,circle,minimum width=0.7cm,label=180:{\normalsize$\{0\}$}] (v2) at (0,1.5) {$h_2$};
            \node[draw,circle,minimum width=0.7cm] (v3) at (1.5,1.5) {};
            \node[draw,circle,minimum width=0.7cm] (v4) at (1.5,0) {$r$};
            
            \node at (0.75,-1) {\textcolor{green!80!black}{\Large\cmark}};
            
            \draw (v1) -- (v2) -- (v3) -- (v4) -- (v1);
        \end{tikzpicture}
        \caption{An instance of the \HRHshort problem from \Cref{ex:ARH} and two possible housings. On the left, we have housing that is not respecting---inhabitant $h_2$ does not accept any refugee in its neighborhood. On the right, the housing is respecting as $h_1$ has one refugee in the neighborhood and $1\in A_{h_1}$. Note that the approval sets of both inhabitants form intervals.}
        \label{fig:ARH:example}
    \end{figure}
	Let the topology be a cycle with four vertices. There are two inhabitants assigned to neighboring vertices. One of these inhabitants, call her~$h_1$, has approval set~$A_{h_1} = \{0,1\}$, and the second one, say~$h_2$, is not approving any refugees in his neighborhood, that is,~$A_{h_2}=\{0\}$. We have~$\RF = \{r\}$. The only valid housing is next to the inhabitant~$h_1$ as housing~$r$ in the neighborhood of~$h_2$ clearly does not respect his preferences. See \Cref{fig:ARH:example} for a more detailed discussion.
\end{example}

As our first result, we observe that even in a very simple settings, it is not guaranteed that any inhabitant-respecting refugees housing exists.

\begin{proposition}\label{lem:SDRH:IR:noInstance}
	There is an instance of the \ARH problem with no inhabitant-respecting refugees housing even if all inhabitants approve intervals.
\end{proposition}

To prove \Cref{lem:SDRH:IR:noInstance}, assume an instance with one inhabitant~$h$ and two refugees~$r_1$ and~$r_2$. Let the topology be~$K^3$, the inhabitant~$h$ be assigned to an arbitrary vertex, and let~$A_h = \{0\}$. There are exactly two possible housings, and in any of them, the inhabitant~$h$ has two neighboring refugees; therefore, there is no inhabitant-respecting housing.

In the previous example, we used the fact that the inhabitant~$h$ does not approve any refugees in his neighborhood. We call such inhabitants \emph{intolerant}. Despite the fact that the instance does not have an inhabitant-respecting housing even if~$A_h = \{1\}$, we observe that intolerant inhabitants can be safely removed.

\begin{proposition}\label{lem:SDRH:IR:intolerantRehousing}
	Let~$\mathcal{I} = (G,I,R,\ia,(A_h)_{h\in I})$ be an instance of the \ARH problem,~$h\in I$ be an inhabitant with~$A_h = \{0\}$, and~$F_h = \{\ia(j)\} \cup U_h$.~$\mathcal{I}$~admits an inhabitant-respecting housing iff the instance~$\mathcal{I}'=(G\setminus F_h,I\setminus\{h\},R,\ia,(A_{h'})_{h'\in I\setminus\{j\}})$ admits an inhabitant-respecting housing.
\end{proposition}
\begin{proof}
	Let~$\mathcal{I}$ be a \yesI, let~$h\in I$ be an inhabitant with~$A_h = \{0\}$, and let~$\pl$ be an inhabitant-respecting refugees housing. Since~$\pl$ is an inhabitant-respecting housing, there is no refugee in the neighborhood of~$h$, so~$\pl$ is a solution even for~$\mathcal{I}'$. 
	
	In the opposite direction, let~$\mathcal{I}'$ be a \yesI and~$\pl'$ be an inhabitant-respecting housing in~$\mathcal{I}'$. As~$\pl'$ houses all refugees to~$V' = V\setminus\{N_G(\ia(h))\}$ the housing~$\pl'$ is also a solution for~$\mathcal{I}$.
\end{proof}

Due to the definition of approval sets, inhabitants without unoccupied neighborhoods are necessarily assumed intolerant and, therefore, can be safely removed by \Cref{lem:SDRH:IR:intolerantRehousing}. Hence, we assume only instances without intolerant inhabitants where every inhabitant has at least one unoccupied vertex in her neighborhood. 

\begin{proposition}\label{lem:SDRH:IR:bipartite}
	Let~$\mathcal{I} = (G,I,R,\ia,(A_h)_{h\in I})$ be an instance of the \ARH problem and~$\{u,v\}\in E(G)$ be an edge such that either~$u,v\in V_\IH$ or~$u,v\in V_U$. Then~$\mathcal{I}$~admits an inhabitant-respecting housing iff the instance~$\mathcal{I}'=((V(G),E(G)\setminus \{\{u,v\}\}),\IH,\RF,\ia,(A_h)_{h\in\IH})$ admits an inhabitant-respecting housing.
\end{proposition}

\Cref{lem:SDRH:IR:bipartite} directly implies that all graphs assumed in this section are naturally bipartite where one part consists of inhabitants and the other side by empty vertices.

We start our investigation of the computational complexity of the \ARHshort problem with a positive result showing that if the topology is a graph of maximum degree~$2$, then we can decide whether inhabitant-respecting housing exists in polynomial-time.

\begin{theorem}\label{lem:SDRH:IR:algo:maxdeg2}
	Every instance of the \ARH problem where the topology is a graph of maximum degree~$2$ can be solved in polynomial time.
\end{theorem}
\begin{proof}
    \newcommand{\DP}{\operatorname{DP}}
	Let~$G$ be a topology and~$C_1,\ldots,C_k$ an arbitrary ordering of its components. First, assume that we have a polynomial-time computable function~$\Call{CanHouse}{C_i,r}$ that returns~$\texttt{true}$ if it is possible to house~$r$ refugees on the component~$C_i$ in the inhabitant-respecting way, and \texttt{false} otherwise. Then, we can simply partition our refugees into~$k$ (possibly empty) sets~$\RF_1,\ldots,\RF_k$ and, if we find an inhabitant-respecting housing~$\pl_i$ of~$\RF_i$ on a component~$C_i$ for every~$i\in[k]$, by joining our partial housings~$\pl_i$ into a general housing~$\pl$, we obtain a solution for the entire instance. An implementation of such an approach is illustrated on \Cref{alg:ARH:P:maxdeg2:global}. The algorithm uses a dynamic programming technique to reduce its running time. Observe that the dynamic programming table is of size at most~$\Oh{k\cdot|\RF|} \in \Oh{n^2}$, and each cell can be computed in time~$\Oh{n \cdot T_{\operatorname{C}}}$, where~$T_{\operatorname{C}}$ is the running time of the~$\textsc{CanHouse}$ function. That is, the overall running time of \Cref{alg:ARH:P:maxdeg2:global} is~$\Oh{n^3\cdot T_{\operatorname{C}}}$ and its correctness is obvious---we check all possible partitions of refugees between the components and return \yes if and only if, for some partition, all components can house its part in an inhabitant-respecting way.
    
    \begin{algorithm}[tb]
        \renewcommand{\DP}{\texttt{T}}
		\caption{A dynamic programming algorithm that decides whether there is an inhabitant-respecting housing of~$R$ refugees.}
		\label{alg:ARH:P:maxdeg2:global}
        \begin{flushleft}
            \textbf{Input}: An instance~$\mathcal{I} = (G,\RF,\IH,\ia,(A_h)_{h\in\IH})$ of \ARHshort.\\
		      \textbf{Output}: \texttt{true} if there is an inhabitant-respecting housing~$\pl$ of~$\RF$ on~$G$, \texttt{false} otherwise.
        \end{flushleft}
		\begin{algorithmic}[1] %
            \State $C_1,C_2,\ldots,C_k \gets \Call{Components}{G}$ \Comment{Fix an ordering of components.}
            \For{$i\in[k]$ \textbf{and}~$j\in[|\RF|]$} \Comment{Initialize the dynamic programming table.}
                    \State~$\DP[i,j] \gets \texttt{undef}$
            \EndFor{}
            \State \Return \Call{SolveRec}{$1$,~$|\RF|$}
            \LComment{A recursive function that checks whether~$r$ refugees can be housed on components~$C_i,\ldots,C_k$.}
            \Function{SolveRec}{$i$,~$r$} 
			\If{$\DP[i,r] = \texttt{undef}$}
                \If{$i = k$}
                    \State~$\DP[i,r] \gets  \Call{CanHouse}{C_i,r}$
                \Else
                    \State~$\DP[i,r] \gets \texttt{false}$
                    \For{~$r' \in [r]$}
                        \If{$\Call{CanHouse}{C_i,r'}$ \textbf{and}~$\Call{SolveRec}{i+1,r-r'}$}
                            \State~$\DP[i,r] \gets \texttt{true}$
                        \EndIf{}
                    \EndFor
                \EndIf{}
            \EndIf{}
            \State \Return~$\DP[i,r]$
            \EndFunction{}
		\end{algorithmic}
	\end{algorithm}
    
    Note that \Cref{alg:ARH:P:maxdeg2:global} is not specific for graphs of maximum degree two and works for an arbitrary topology where the function \textsc{CanHouse} can be implemented in polynomial time for each of its components. Therefore, for our proof, it remains to show that this is indeed the case for topologies of maximum degree two.

    Let~$G$ be a connected topology of maximum degree two. Then,~$G$ is either a path or a cycle~\cite{Diestel2017}. Our algorithm for deciding whether an inhabitant-respecting housing exists is based on the dynamic programming approach combined with the gradual elimination of inhabitants' approval sets and exhaustive application of \Cref{lem:SDRH:IR:intolerantRehousing}. We first introduce an algorithm that solves the problem on a path, and then we show how to tweak the algorithm to solve also cycles.
	
	Let the topology be a path~$P=v_1v_2\ldots v_k$,~$k \geq 3$. For every vertex~$v_i\in V(P)$, we have a dynamic programming table~$\operatorname{DP}_{i}[f_\ell,f_c,\rho]$ that stores either \texttt{true} or \texttt{false} based on whether we can house~$\rho$ refugees on the sub-path~$v_i,\ldots,v_k$, subject to binary flags~$f_\ell$ and~$f_c$, where the flag~$f_\ell$ applies when we are handling an occupied vertex~$v_i$ and is set to~$0$ (to~$1$, respectively) if the vertex~$v_{i-1}$ is not used (is used) for housing in a solution. The flag~$f_c$, on the other hand, applies for empty vertices and is set to~$1$ if~$v_i$ is used for housing and to~$0$ if~$v_i$ is not housed in a solution.  
	Without loss of generality, we can assume that the vertices~$v_1$ and~$v_k$ are occupied by inhabitants; if it is not the case, we can attach to~$v_1$ (and similarly to~$v_k$) a pendant vertex~$v_0$ and make it occupied by an inhabitant~$h$ with~$A_{h} = \{0,1\}$. Since this newly added inhabitant~$h$ approves any neighborhood, he clearly does not affect the solution.
	
	First, we show how to compute the table for the vertex~$v_k$. Let~$h$ be an inhabitant such that~$\ia(h) = v_k$. By \Cref{lem:SDRH:IR:noInstance}, we have that~$h$ is not intolerant. Consequently, the approval set~$A_h$ is either~$\{1\}$ or~$\{0,1\}$. In any case, since~$v_k$ is a leaf, the sub-path~$v_k$ cannot house a refugee regardless of whether~$v_{k-1}$ is occupied or not. Moreover, by the approval set of~$h$,~$h$ does not care whether~$v_{k-1}$ is used for housing or not; we can set~$\operatorname{DP}_{k}[f_\ell,f_c,0] = \texttt{true}$ for all~$f_\ell,f_c\in\{0,1\}$ and~$\operatorname{DP}_{k}[f_\ell,f_c,\rho] = \texttt{false}$ otherwise.
	
	Next, the computation for empty vertices is relatively easy. Let~$v_i$,~$i\in[2,k-1]$, be an empty vertex. We set 
	\[
		\DP_{i}[f_\ell,f_c,\rho] = \begin{cases}
			\texttt{false} & \text{if~$\rho = 0$ and~$f_c = 1$, and}\\
			\DP_{{i+1}}[f_c,0,\rho-f_c] & \text{otherwise.}
		\end{cases}
	\]
	
	Finally, for a vertex~$v_i$,~$i\in[3,k-2]$, occupied by an inhabitant~$h$, the computation of the dynamic programming table varies based on the approval set of~$h$. By \Cref{lem:SDRH:IR:noInstance}, the inhabitant~$h$ is clearly not intolerant. Now, we show how the computation works for each approval set separately.
	
	\textbf{Case 1.~$\mathbf{A_h = \{0,1,2\}}$:} In this case, the inhabitant approves any neighborhood. Hence, we can compute the maximum value simply by setting~$\DP_{i}[f_\ell,f_c,\rho] = \DP_{{i+1}}[0,0,\rho]\lor\DP_{{i+1}}[0,1,\rho]$.
	
	\textbf{Case 2.~$\mathbf{A_h = \{1\}}$:} In this case, the inhabitant~$h$ requires exactly one refugee in the neighborhood. Hence, we handle the cases differently based on the occupancy of~$v_{i-1}$, which is captured in the flag~$f_\ell$. Formally, we set
	$\DP_{i}[f_\ell,f_c,\rho] = \DP_{{i+1}}[0,1-f_\ell,\rho]$.
	
	\textbf{Case 3.~$\mathbf{A_h = \{2\}}$:} The inhabitant~$h$ requires both neighbors to be occupied. Hence, if the vertex~$v_{i-1}$ is not part of the solution, we need to refuse the solution as invalid. Otherwise, the vertex~$v_{i+1}$ also needs to be housed in every solution, leading to the following:
	\[
		\DP_{i}[f_\ell,f_c,\rho] = \begin{cases}
			\DP_{{i+1}}[ 0, 1, \rho ] & \text{if~$f_\ell = 1$, and}\\
			\texttt{false} & \text{otherwise.}
		\end{cases}
	\]
	
	\textbf{Case 4.~$\mathbf{A_h = \{0,1\}}$:} We need to secure that at most one of~$v_{i-1}$ and~$v_{i+1}$ is part of a housing maximizing the number of housed refugees. Hence, if~$f_\ell = 1$, the vertex~$v_{i+1}$ cannot be 
	\[
		\DP_{i}[f_\ell,f_c,\rho] = \begin{cases}
			\DP_{{i+1}}[0,0,\rho]& \text{if~$f_\ell = 1$, and}\\
			\DP_{{i+1}}[0,0,\rho]\lor\DP_{{i+1}}[0,1,\rho] & \text{otherwise.}
		\end{cases}
	\]
	
	\textbf{Case 5.~$\mathbf{A_h = \{0,2\}}$:} We need to ensure that either none or both neighbors are occupied. Formally, we set~$\DP_{i}[f_\ell,f_c,\rho] = \DP_{{i+1}}[0,f_\ell,\rho]$.
	
	\textbf{Case 6.~$\mathbf{A_h = \{1,2\}}$:} This last case is symmetric to the case 4; if~$f_\ell = 1$, then the vertex~$v_{i+1}$ may or may not be housed, while if~$f_\ell = 0$, the vertex~$v_{i+1}$ is necessarily part of a solution housing. Formally, we set
	\[
	\DP_{i}[f_\ell,f_c,\rho] = \begin{cases}
		\DP_{{i+1}}[0,1,\rho]& \text{if~$f_\ell = 0$, and}\\
		\DP_{{i+1}}[0,0,\rho]\lor\DP_{{i+1}}[0,1,\rho] & \text{otherwise.}
	\end{cases}
	\]
	
	Once the dynamic programming table for every vertex~$v$ is correctly calculated, we can check whether~$\DP_{1}[0,0,|\RF|]$ is set to \texttt{true}. If this is the case, the path can house all refugees, and the algorithm returns \yes. Otherwise, the algorithm returns \no. The size of the dynamic table for a single vertex~$v$ is~$\Oh{n}$ and each cell can be computed in constant time. Hence, the overall running time of the algorithm is~$\Oh{n^2}$.
    
	If the graph~$G$ is a cycle, there is at least one vertex occupied by an inhabitant. Without loss of generality, let~$v_1$ be occupied by an inhabitant~$h\in I$. The idea of the algorithm is, based on the approval set of the inhabitant~$h$, to try to add all possible neighbors to a solution housing, remove the vertex~$v_1$ together with its neighbors from~$G$, update the approval sets of the inhabitants in the second neighborhood of~$v_1$, and use the previous algorithm for paths to decide the reduced instance.

    Regardless of whether~$G$ is a path or a cycle, we can verify whether it is possible to house~$|\RF|$ refugees on~$G$ in~$\Oh{n^2}$ time. From the running time analysis of \Cref{alg:ARH:P:maxdeg2:global}, we obtain that there is an algorithm that solves \ARHshort problem on graphs of the maximum degree at most two in~$\Oh{n^3\cdot n^2} = \Oh{n^5}$ time.
\end{proof}

Unfortunately, as the following theorem shows, the bounded-degree condition from \Cref{lem:SDRH:IR:algo:maxdeg2} cannot be relaxed anymore.

\begin{theorem}\label{thm:ARH:NPh}
	The \ARH problem is \NPc even if the topology is a graph of maximum degree~$3$ and all inhabitants approve intervals.
\end{theorem}
\begin{proof}
	Given a housing~$\pl$, it is easy to verify in polynomial time whether~$\pl$ is inhabitant-respecting by enumerating all inhabitants and comparing their neighborhoods with approval lists. Thus, \ARHshort is indeed in \NP.
	
	For \NPhness, we present a polynomial-time reduction from a variant of the \probName{$2$-Balanced~$3$-SAT} problem, which is known to be \NPc~\cite{Tovey1984,FialaGK2005,BermanKS2003}. In this variant of \probName{$3$-SAT}, we are given a propositional formula~$\varphi$ with~$\xi$ variables~$x_1,\ldots,x_\xi$ and~$\mu$ clauses~$C_1,\ldots,C_\mu$ such that each clause contains at most~$3$ literals and every variable appears in at most~$4$ clauses -- at most twice as a positive literal and at most twice as a negative literal. Later in this paper, we will refer to this reduction as \emph{basic reduction}.
	
	We construct an equivalent instance~$\mathcal{I}$ of \ARHshort as follows. We represent every variable~$x_i$,~$i\in[\xi]$, by a single \emph{variable gadget}~$X_i$ which is a path~$t_iv_if_i$. The vertex~$v_i$ is occupied by an inhabitant~$g_i$, called \emph{variable-guard}, with approval set~$\{1\}$. All other vertices are empty, and we call the vertex~$t_i$ the \emph{$t$-port} and the vertex~$f_i$ the \emph{$f$-port}. Every \emph{clause}~$C_j$,~$j\in\mu$, is represented by a single vertex~$c_j$ occupied by an inhabitant~$h_j$, called \emph{clause-guard}, who approves the interval~$[1,|C_j|]$ and is connected to the~$t$-port of the variable gadget~$X_i$ if the variable~$x_i$ occurs as a positive literal in~$C_j$ and to the~$f$-port of~$X_i$ if~$x_i$ occurs as a negative literal in~$C_j$. To complete the reduction, we set~$|\RF|=\{r_1,\ldots,r_\xi\}$.
	
	For the correctness of the construction, let~$\varphi$ be a satisfiable \mbox{\probName{$2$-Balanced~$3$-SAT}} formula and~$\alpha$ be a truth assignment. For every variable~$x_i$, we assign the refugee~$r_i$ to~$t_i$ if~$\alpha(x_i) = 1$ and to~$f_i$ if~${\alpha(x_i) = 0}$, respectively. This housing is clearly a solution of~$\mathcal{I}$ since every variable-guard has exactly one refugee in the neighborhood and every clause-guard~$h_j$ has at least one refugee in the neighborhood since~$\alpha$ satisfies all clauses.
	
	In the opposite direction, observe that due to variable-guards, there is exactly one refugee assigned to every variable gadget and, thus, in every assignment~$\pi$ there is no variable gadget~$X_i$ such that the~$t$-port and the~$f$-port are occupied at the same time. Hence, we can set~$\alpha(x_i)$ equal to~$1$ if and only if the~$t$-port is occupied by a refugee. Clearly,~$\alpha$ is a truth assignment as~$\pi$ has to satisfy each inhabitant occupying clause vertex.
	
	By definition, every clause contains at most~$3$ literals, and thus the degree of every vertex~$c_j$,~${j\in[\mu]}$, is at most~$3$. For every variable gadget~$X_i$,~$i\in[\xi]$, the vertex~$v_i$ has degree~$2$ and both~$t$-port and~$f$-port have degree at most~$3$ -- they are adjacent to~$v_i$ and at most two vertices representing clauses. Hence, the bounded-degree condition holds, and the construction can be clearly done in polynomial time, finishing the proof.
\end{proof}

Since the above results clearly show that the problem is computationally hard even in simple settings, we turn our attention to the parameterized complexity of the \ARHshort problem. In particular, we study the problem's complexity from the viewpoint of natural parameters, such as the number of refugees, the number of inhabitants, the number of empty vertices, and various structural parameters restricting the shape of the topology.

We start with instances, where the number of refugees to house is small. It turns out that even with such a strong restriction, one cannot expect fixed-parameter tractability.

\begin{theorem}\label{thm:ARH:Wh:R}
	The \ARH problem is \Wh[2] parameterized by the number of refugees~$|R|$ even if all inhabitants approve intervals.
\end{theorem}
\begin{proof}
	We reduce from the \probName{Dominating Set} problem, which is known to be \Wc[2] when parameterized by the solution is size~$k$~\cite{DowneyF1995}. 
	The instance~$\mathcal{I}$ of \probName{Dominating Set} consists of a simple undirected graph~$H$ and an integer~$k\in\N$. The goal is to decide whether there is a set~$D\subseteq V(H)$ of size at most~$k$ such that each vertex~$v\in V(H)$ is either in~$D$ or at least one of its neighbors is in~$D$.
	
	We construct an equivalent instance~$\mathcal{J}$ of the \ARHshort problem as follows. We start by defining the topology~$G$. For each vertex~${v\in V(H)}$ we add two vertices~$\ell_v$ and~$p_v$. The vertex~$\ell_v$ represents the original vertex and is intended to be free for refugees. The vertex~$p_v$ is occupied by an inhabitant~$h_v$ with~$A_{h_v} = [1,|N_H[v]|]$. This inhabitant ensures that there is at least one refugee housed in the closed neighborhood of~$p_v$. The edge set of the topology~$G$ is~$\bigcup_{v\in V}\{\{p_v,\ell_w\}\mid w\in N_H[v]\}$. To complete the construction, we set~$|R|=\{r_1,\ldots,r_k\}$.
	For an overview of our construction, we refer the reader to \Cref{fig:SDRH:IR:Wh:refugees}.
	
	\begin{figure}[tb]
        \centering
		\begin{tikzpicture}
			\node[draw,circle,inner sep=1.5pt,label=90:{\small$ \ell_1$}] (v1) at (0,0) {};
			\node[draw,circle,inner sep=1.5pt,label=90:{\small$ \ell_2$}] (v2) at (1,0) {};
			\node[draw,circle,inner sep=1.5pt,label=90:{\small$ \ell_3$}] (v3) at (2,0) {};
			\node[draw,circle,inner sep=1.5pt,label=90:{\small$ \ell_4$}] (v4) at (3,0) {};
			\node at (4,0) {$\dots$};
			\node[draw,circle,inner sep=1.5pt,label=90:{\small$ \ell_n$}] (v5) at (5,0) {};
			
			\node[white,fill=black,circle,inner sep=1.5pt,label=270:{\small$[1,3]$}] (n1) at (0,-1.5) {$h_1$};
			\node[white,fill=black,circle,inner sep=1.5pt,label=270:{\small$[1,2]$}] (n2) at (1,-1.5) {$h_2$};
			\node[white,fill=black,circle,inner sep=1.5pt,label=270:{\small$[1,4]$}] (n3) at (2,-1.5) {$h_3$};
			\node[white,fill=black,circle,inner sep=1.5pt,label=270:{\small$[1,3]$}] (n4) at (3,-1.5) {$h_4$};
			\node at (4,-1.5) {$\dots$};
			\node[white,fill=black,circle,inner sep=1.5pt,label=270:{\small$[1,3]$}] (n5) at (5,-1.5) {$h_n$};
			
			\draw (n1) edge (v1) edge (v3) edge (v4);
			\draw (n2) edge (v2) edge (v5);
			\draw (n3) edge (v3) edge (v1) edge (v4) edge (v5);
			\draw (n4) edge (v4) edge (v1) edge (v3);
			\draw (n5) edge (v5) edge (v2) edge (v3);
		\end{tikzpicture}
		\caption{An illustration of the construction used in the proof of \Cref{thm:ARH:Wh:R}.}
		\label{fig:SDRH:IR:Wh:refugees}
	\end{figure}
	
	Let~$\mathcal{I}$ be a \yesI and~$D$ be a dominating set of size~$k$. For every vertex~$v\in D$, we house a refugee in the vertex~$\ell_v$. Since~$D$ was a dominating set of size~$k$, in the closed neighborhood of every~$v\in V$ in~$G$, there is at least one vertex~$u\in D$. Therefore, for every inhabitant~$h_v$, there is at least one refugee in his neighborhood, and~$\mathcal{J}$ is indeed a \yesI.
	
	In the opposite direction, let~$\mathcal{J}$ be a \yesI and~$\pl$ be a solution housing. We set~$D$ to be~$\{v\in V(H)\mid \exists h\in\IH\colon \pl(h)=v\}$. Due to the definition of approved intervals of the inhabitants, it holds for every~$v\in V(H)$ either~$v$ or at least one of his neighbors is in~$D$, as otherwise the inhabitant~$h_v$ would not be respected.
	
	To complete the proof, we recall that~$|\RF|=k$ and, hence, the presented reduction is indeed a parameterized reduction.
\end{proof}

We complement \Cref{thm:ARH:Wh:R} with an algorithm that runs in time that matches the lower bound given in this theorem. The following result, however, shows that the problem can be solved in polynomial time for a constant number of refugees.

\begin{theorem}\label{thm:ARH:XP:R}
	The \ARH problem can be solved in~$n^{\Oh{|R|}}$ time. That is, \ARH is in \XP parameterized by the number of refugees. Moreover, unless ETH fails, there is no algorithm that solves \ARH in~$f(|R|)\cdot n^{o(|R|)}$ time for any computable function~$f$.
\end{theorem}
\begin{proof}
	Our algorithm is a simple brute-force. Let~$V_U = V(G)\setminus V_I$ be the number of empty vertices and let~$n = |V|$. Note that~$|V_U| \leq n$. We try all subsets of~$V_U$ of size~$|\RF|$, and for each such subset, we check in linear time whether the housing (note that it does not matter which refugees are housed where) is inhabitant-respecting. If at least one subset leads to an inhabitant-respecting housing, we return \yes. Otherwise, we return \no. As there are~$|V_U|^\Oh{|\RF|} = n^\Oh{|\RF|}$ such subsets and each can be verified in polynomial time, the total running time is~$n^\Oh{|\RF|}$.
	
	For the running time lower-bound, recall that according to \Cref{thm:ETH:XP}, the \probName{Dominating Set} problem cannot be solved in~$f(k)\cdot |\mathcal{I}|^{o(k)}$ time for any computable function~$f$, unless ETH fails. Assume that there is an algorithm~$\mathcal{A}$, that solves \ARHshort in~$f(|\RF|)\cdot n^{o(|\RF|)}$ time. Then we can reduce an instance~$\mathcal{I}$ of the \probName{Dominating Set} problem to an equivalent instance of the \ARHshort problem using the construction from \Cref{thm:ARH:Wh:R}, solve the reduced instance using algorithm~$\mathcal{A}$, and return the same response for~$\mathcal{I}$. As the construction use~$|R|=k$, this is an algorithm for \probName{Dominating Set} running in~$f(k)\cdot |\mathcal{I}|^{o(k)}$ time, which contradicts \Cref{thm:ETH:XP}.
\end{proof}

As the number of refugees is not a parameter promising tractable algorithm, even if all inhabitants approve intervals, we focus on the case where the number of inhabitants is small. 

The first result is an \FPT algorithm for this parameterization under the assumption that each inhabitant approves the interval. Our algorithm is based on integer linear programming formulation of the problem, and we use the following result of \citet{EisenbrandW2018}.

\begin{theorem}[{\cite[Theorem 2.2]{EisenbrandW2018}}]\label{thm:ILP:FPT:constraints}
	Integer linear program~${\mathbb{A}x\leq b}$,~$x \geq 0$, with~$n$ variables and~$m$ constraints can be solved in~$$(m\Delta)^\Oh{m}\cdot||b||_\infty^2$$ time, where~$\Delta$ is an upper-bound on all absolute values in~$\mathbb{A}$.  
\end{theorem}

\begin{theorem}\label{thm:ARH:IR:FPT:I:intervals}
	If all inhabitants approve intervals, then the \ARH problem can be solved in~$2^\Oh{|\IH|\cdot\log|\IH|}$ time. That is, \ARH is in \FPT when parameterized by the number of inhabitants~$|\IH|$.
\end{theorem}
\begin{proof}
	We solve the \ARHshort problem using an integer linear programming formulation of the problem. We introduce one binary variable~$x_v$ for every empty vertex~$v\in V_U$ representing if a refugee is housed on~$v$ or not. Next, we add the following constraints.
	\begin{align}
		\forall h \in \IH\colon \sum_{v\in N_G(\ia(h))} x_v &\geq \min A_h \label{eq:lb}\\
		\forall h \in \IH\colon \sum_{v\in N_G(\ia(h))} x_v &\leq \max A_h \label{eq:ub}\\
		\sum_{v\in V_U} x_v &= |\RF|.\label{eq:refnum}
	\end{align}
	
	\Cref{eq:lb,eq:ub} ensure that the number of refugees in the neighborhood of each inhabitant is in its approved interval, while \Cref{eq:refnum} secures that all refugees are housed somewhere. Using \Cref{thm:ILP:FPT:constraints}, we see that the given integer program can be solved in time~$|\IH|^\Oh{|\IH|}\cdot n^\Oh{1} = 2^\Oh{|\IH|\cdot\log|\IH|}\cdot n^\Oh{1}$, as~$m = 2|\IH|+1$,~$\Delta = 1$, and~$||b||_\infty \leq n$. That is, \ARHshort is in \FPT parameterized by the number of inhabitants~$|\IH|$.
\end{proof}

Note that it would be possible to provide a different ILP formulation of the problem and use the famous theorem of \citet{Lenstra1983} to show membership in \FPT; however, this would yield an algorithm with much worse (i.\,e., doubly-exponential) running-time.

As we will show later, the result from \Cref{thm:ARH:IR:FPT:I:intervals} cannot be easily generalized to the case with inhabitants approving general sets. However, we can show that if the number of intervals in each approval set is bounded, the problem is still fixed-parameter tractable.

\begin{theorem}\label{thm:ARH:IR:FPT:I:delta}
	The \ARH problem is fixed-parameter tractable when parameterized by the combined parameter the number of inhabitants~$|\IH|$ and the maximum number of disjoint intervals~$\delta$ in the approval sets.
\end{theorem}
\begin{proof}
	The basic idea of the algorithm is to guess for each inhabitant~$h\in\IH$ his or her \emph{effective interval} and then use \Cref{thm:ARH:IR:FPT:I:intervals} to decide whether the guess is correct. More formally, it is given that each inhabitant's~$h\in\IH$ approval set~$A_h$ consists of at most~$\delta$ disjoint intervals. Therefore, we guess for each inhabitant her effective interval~$E_h$ and use it as an approval set for~$h$. If we assume only~$E_h$ for each~$h\in\IH$, then every inhabitant approves only an interval and we can use \Cref{thm:ARH:IR:FPT:I:intervals} to decide in time~$2^\Oh{|\IH|\cdot\log|\IH|}\cdot n^\Oh{1}$ whether our guess is correct. If at least one guess is correct, we return \yes. Otherwise, the result is \no. There are~$\delta^\Oh{|\IH|}\in 2^\Oh{|\IH|\cdot\log(\delta)}$ possible guesses and each guess can be verified in~$2^\Oh{|\IH|\cdot\log|\IH|}\cdot n^\Oh{1}$ time. Therefore, the overall running time is~$2^\Oh{|\IH|\cdot\log(\delta)}\cdot2^\Oh{|\IH|\cdot\log|\IH|}\cdot n^\Oh{1}$, which is clearly in \FPT.
\end{proof}

Now, we show that the parameter~$\delta$ from \Cref{thm:ARH:IR:FPT:I:delta} cannot be dropped while keeping the problem tractable.

\begin{theorem}\label{thm:ARH:Wh:I}
	The \ARH problem is \Wh parameterized by the number of inhabitants~$|\IH|$.
\end{theorem}
\begin{proof}
	Our reduction is similar to the one of \citet[Theorem 3.1]{KnopKMT2019}. We reduce from the \probName{Multicolored Clique} problem where we are given a~$k$-partite graph~$G=(V_1\cup\cdots\cup V_k,E)$ and the goal is to find a complete subgraph with~$k$ vertices such that it contains a vertex from every~$V_i$,~$i\in[k]$. \probName{Multicolored Clique} is known to be \Wh with respect to~$k$~\cite{FellowsHRV2009}. We may assume that every color class~$V_i$,~$i\in[k]$, is of size~$n$. By~$E_{i,j}$ we denote the set of edges between color classes~$V_i$ and~$V_j$, that is,~${E_{i,j} = \{\{u,v\}\mid u\in V_i \land v\in V_j\}}$ and we may assume that for every pair of distinct~$i,j\in[k]$ we have~$|E_{i,j}|=m$.
	
	We begin the reduction by fixing a bijection~$\nu_i\colon V_i\to[n]$ for every~$i\in[k]$ and a bijection~$\epsilon_{i,j}\colon E_{i,j}\to[m]$ for every pair of distinct~$i,j\in[k]$. Next, for every color class~$V_i$, we add a \emph{vertex-selection gadget}~$S_i$ consisting of~$n$ vertices and leave these vertices empty. The number of refugees housed on the vertices of~$S_i$ will correspond to a vertex in~$V_i$ that is part of the clique. Then, for every~$E_{i,j}$,~$i,j\in[k]$, we add a set~$T_{i,j}$ with~$m\cdot n^2$ empty vertices and connect these to the vertex~$M_{i,j}$. The vertex~$M_{i,j}$ is occupied by an inhabitant with approval set~$\{tn^2\mid t\in[m]\}$ and we call~$T_{i,j}$ an \emph{edge-selection gadget}. Similarly to the vertex selection gadget, the number of refugees assigned to~$T_{i,j}$ will correspond to the edge selected for the solution. To ensure that the choice performed in the vertex-selection gadgets and the edge-selection gadgets is compatible, we introduce two vertices~$G_{i,j}$ and~$G_{j,i}$ for every~$E_{i,j}$,~$i,j\in[k]$ and~$i\not= j$. The vertex~$G_{i,j}$ is occupied by an inhabitant with approval set~$\{\epsilon_{i,j}(e)\cdot n^2 + \nu_i(v) \mid e\in E_{i,j} \land v\in e \land v\in V_i\}$ and is adjacent to every vertex in~$S_i$, and the vertex~$G_{j,i}$ is occupied by an inhabitant with approval set~$\{\epsilon_{i,j}(e)\cdot n^2 + \nu_j(v) \mid e\in E_{i,j} \land v\in e \land v\in V_j\}$ and is adjacent to every vertex in~$S_j$. To complete the construction, we set~$|R| = \binom{k}{2}\cdot m\cdot n^2 + k\cdot n$ and introduce the same number of auxiliary vertices of degree~$0$ which are intended for the remaining refugees not assigned to vertex- and edge-selection gadgets. For an overview of our construction, please refer to \Cref{fig:ARH:IR:Wh:I}.
	\begin{figure}[bt!]
        \centering
		\begin{tikzpicture}
			\node[circle,inner sep=1pt,white,fill=black] (Gij) at (0,0.8) {\small$G_{i,j}$};
			\node[circle,inner sep=1pt,white,fill=black] (Gji) at (0,-0.8) {\small$G_{j,i}$};
			
			\node[circle,inner sep=1pt,white,fill=black] (Mij) at (3,0) {\small$M_{i,j}$};
			
			\node at (-1.9,1.8) {\small$S_i$};
			\draw[rounded corners=2pt,fill=gray!15] (-1.7,2) rectangle (-1.3,0.4);
			\foreach[count=\i] \y in {1.8,1.4,1,0.6}{
				\node[draw,circle,inner sep=1pt,fill=white] (v\i) at (-1.5,\y) {};
				\draw (v\i) -- (Gij);
			}

			\node at (-1.9,-1.8) {\small$S_j$};
			\draw[rounded corners=2pt,fill=gray!15] (-1.7,-2) rectangle (-1.3,-0.4);
			\foreach[count=\i] \y in {-1.8,-1.4,-1,-0.6}{
				\node[draw,circle,inner sep=1pt,fill=white] (u\i) at (-1.5,\y) {};
				\draw (u\i) -- (Gji);
			}
			
			\node at (2.1,1.8) {\small$T_{i,j}$};
			\draw[rounded corners=2pt,fill=gray!15] (1.8,2) rectangle (1.2,-2);
			\foreach[count=\i] \y in {1.8,1.4,1,0.6,0.2,-0.2,-0.6,-1,-1.4,-1.8}{
				\node[draw,circle,inner sep=1pt,fill=white] (e\i) at (1.5,\y) {};
				\draw (e\i) -- (Gij);
				\draw (e\i) -- (Gji);
				\draw (e\i) -- (Mij);
			}
		\end{tikzpicture}
		\caption{An overview of the construction used in the proof of \Cref{thm:ARH:Wh:I}. The sets~$S_i$,~$S_j$, and~$T_{i,j}$ consist of unoccupied vertices, and there are sequentially~$n$,~$n$, and~$mn^2$ of them in each set. The vertex~$G_{i,j}$ is occupied by an inhabitant with an approval set~$\{\epsilon_{i,j}(\{u,v\})\cdot n^2 + \nu_i(u) \mid \{u,v\}\in E_{i,j} \land u \in V_i\}$ and the vertex~$G_{j,i}$ is occupied by an inhabitant with an approval set~$\{\epsilon_{i,j}(\{u,v\})\cdot n^2 + \nu_j(v) \mid \{u,v\}\in E_{i,j} \land v \in V_j\}$. The inhabitant occupying the vertex~$M_{i,j}$ approves the set~$\{t\cdot n^2\mid t\in[m]\}$.}
		\label{fig:ARH:IR:Wh:I}
	\end{figure}
    
	For correctness, let~$\mathcal{I} = (G,k)$ be a \yesI and~${v_i,\ldots,v_k}$, where~$v_i\in V_i$, be vertices that form a clique in~$G$. For every~$i\in[k]$ we assign~$\nu_i(v_i)$ refugees to empty vertices of~$S_i$ and for each pair of distinct~$i,j\in[k]$, we house~$\epsilon_{i,j}(\{v_i,v_k\})\cdot n^2$ refugees on empty vertices of~$T_{i,j}$. The remaining refugees are assigned to the auxiliary vertices. The only inhabitants occupy the vertices~$M_{i,j}$,~$G_{i,j}$, and~$G_{j,i}$. The inhabitants of~$M_{i,j}$ are easily satisfied since we assign some multiple of~$n^2$ to every~$T_{i,j}$. An inhabitant in~$G_{i,j}$ is adjacent to~${\epsilon_{i,j}(\{v_i,v_k\})\cdot n^2}$ refugees from~$T_{i,j}$ and~${\nu_i(v_i)}$ refugees from~$S_i$ while~$G_{j,i}$ is adjacent to~${\epsilon_{i,j}(\{v_i,v_k\})\cdot n^2}$ refugees from~$T_{i,j}$ and~${\nu_j(v_j)}$ refugees from~$S_j$, respectively. This complies with their approval set.
	
	In the opposite direction, let the equivalent \ARHshort instance~$\mathcal{I'}$ be a \yesI and~$\pl$ be an inhabitant respecting housing in~$\mathcal{I'}$. Due to the inhabitants of~$M_{i,j}$, where~$i,j\in[k]$, there is some positive multiple of~$n^2$ refugees assigned to every~$T_{i,j}$ that corresponds to some edge~$e\in E_{i,j}$ in the original graph. Moreover, due to inhabitant on~$G_{i,j}$, the number of refugees assigned to~$S_i$ corresponds to the identification of some vertex~$v\in V_i$ that is necessarily incident to~$e$. The same holds for~$G_{j,i}$.
	
	It is not difficult to see that the vertices~$G_{i,j}$ and~$G_{j,i}$ together with~$M_{i,j}$, where~$i,j\in[k]$ and~$i\not= j$, are the only vertices occupied by the inhabitants. Consequently, the number of inhabitants is~$\binom{k}{2} + k\cdot(k-1) = \Oh{k^2}$. Therefore, the reduction is indeed a parameterized reduction, finishing the proof.
\end{proof}

Again, we complement the hardness lower bound given in the previous theorem with a matching algorithmic upper bound.

\begin{theorem}\label{thm:ARH:XP:I}
	The \ARH problem can be solved in~$2^\Oh{|\IH|\cdot\log|\IH|}\cdot n^\Oh{|\IH|}$ time, that is, \ARHshort is in \XP when parameterized by the number of inhabitants~$|\IH|$. Moreover, unless ETH fails, there is no algorithm that solves \ARHshort in~$f(|\IH|)\cdot n^{o(\sqrt{|\IH|})}$ time for any computable function~$f$.
\end{theorem}
\begin{proof}
	For every inhabitant, we guess the number of refugees in his neighborhood in a hypothetical solution~$\pl'$. For each such guess, we run the algorithm from \Cref{thm:ARH:IR:FPT:I:intervals} to verify whether such a housing can be realized. Overall, we create~$n^\Oh{|\IH|}$ ILP instances, and each instance can be decided in \FPT time. As we check all possible solutions, the algorithm is trivially correct.

    For the running time lower-bound, recall that by \Cref{thm:ETH:XP}, there is no algorithm solving \probName{Multicolored Clique} in~$f(k)\cdot n^{o(k)}$ time. For the sake of contradiction, assume that there exists an algorithm~$\mathbb{A}$ solving \ARHshort in~$g(|\IH|)\cdot n^{o(\sqrt{|\IH|})}$ time. Then, given an instance~$\mathcal{I}$ of \probName{Multicolored Clique}, we can, in polynomial time, construct an equivalent instance~$\mathcal{I'}$ of \ARHshort as in \Cref{thm:ARH:Wh:I}, decide it using algorithm~$\mathbb{A}$, and return the same outcome for~$\mathcal{I}$. Overall, we obtain an algorithm running deciding \probName{Multicolored Clique} in~$g(|\IH|)\cdot n^{o(\sqrt{|\IH|})} + n^\Oh{1} = g(k^2)\cdot n^{o(\sqrt{k^2})} = f(k)\cdot n^{o(k)}$ time, which contradicts \Cref{thm:ETH:XP}.
\end{proof}

It is easy to see that the previous \XP algorithm becomes fixed-parameter tractable if we additionally parameterize by the largest number of refugees approved. Formally, we have the following.

\begin{corollary}
	The \ARH problem is fixed-parameter tractable when parameterized by the number of inhabitants~$|\IH|$ and the largest approved number~$\max_{h\in\IH}\max A_h$ combined.
\end{corollary}

By careful guessing, we can prove that a fixed-parameter tractable algorithm exists for the number of refugees and the number of inhabitants, combined.

\begin{theorem}\label{thm:ARH:FPT:R+I}
	The \ARH problem is fixed-parameter tractable when parameterized by the number of refugees~$|\RF|$ and the number of inhabitants~$|\IH|$, combined.
\end{theorem}
\begin{proof}
    We fix an arbitrary ordering~$(r_1,\ldots,r_{|\RF|})$ of the refugees. Now, we try all possible vectors~$\vec{s}=(\IH_{1},\ldots,\IH_{|\RF|})$, where~$\IH_j \subseteq V_{\IH}$ for every~$j\in[|\RF|]$. Intuitively, this vector represents, for each refugee~$r\in\RF$, its neighborhood in a hypothetical solution. For every such~$\vec{s}$, we verify whether this vector describes a valid housing. To do so, we first create a temporal topology~$G' = G$. Then, we traverse~$\vec{s}$ from left to right, and, for every~$j$, we check whether~$|\bigcap_{v \in \IH_j} N_{G'}(v) | \geq 1$. If the intersection is empty, we refuse the current vector~$\vec{s}$. Otherwise, we remove arbitrary vertex~$w$, which is an element of this intersection, from~$G'$ (that is, we set~$G' = G'\setminus\{w\}$), set~$\pl(r_j) = w$, and continue with another refugee. When the algorithm constructs the whole housing~$\pl$, we need to check that~$\pl$ is inhabitant-respecting. This can be easily done by enumerating all inhabitants and checking that, in the topology~$G$, they approve their neighborhoods with respect to the housing~$\pl$. If this is the case for some vector~$\vec{s}$, the algorithm returns \yes.
\end{proof}

The last assumed natural parameter is the number of empty vertices~$|V_U|$ the refugees can be housed on. Note that~$|V_U| \geq |\RF|$. This parameterization yields, in contrast to \Cref{thm:ARH:Wh:R}, a simple algorithm running in \FPT time, which is, despite its simplicity, optimal assuming the Exponential Time Hypothesis. 

\begin{theorem}\label{thm:ARH:FPT:empty}
	The \ARH problem can be solved in~$2^\Oh{|V_U|}\cdot n^\Oh{1}$ time, that is, \ARH is in \FPT when parameterized by the number of empty vertices~$|V_U|$. Moreover, unless ETH fails, there is no algorithm solving \ARH in~${2^{o(|V_U|)}\cdot n^\Oh{1}}$ time even if all inhabitants approve intervals.
\end{theorem}
\begin{proof}
	First, observe that if~$|V_U| \leq |\RF|$, then the instance is trivially \noI, as it is not possible to house all refugees. Hence,  we can enumerate all subsets of~$\pl\subseteq V_U$ of size~$|\RF|$ and, for every such~$\pl$, check whether it is an inhabitant-respecting housing. If at least one~$\pl$ is inhabitant-respecting, we return \yes. Otherwise, the algorithm outputs \no. As the algorithm tries all possible housings, it is trivially correct. Moreover, there are~$2^\Oh{|V_U|}$ different housings~$\pl$, and each of them can be verified in polynomial time. Hence, \ARHshort is in \FPT with respect to the number of empty vertices~$|V_U|$.
    
	For the running time lower-bound, suppose that there is an algorithm~$\mathcal{A}$ that solves \ARHshort in~$2^{o(|V_U|)}\cdot n^\Oh{1}$. Then, for every \probName{$2$-Balanced~$3$-SAT} formula~$\varphi$, we can use the basic reduction from \Cref{thm:ARH:NPh}, which can be clearly done in polynomial time, to create an equivalent instance~$\mathcal{I}$ of the \ARHshort problem, solve~$\mathcal{I}$ in~$2^{o(|V_U|)}\cdot n^\Oh{1}$ time, and then reconstruct a solution for~$\varphi$. Overall, this gives us an algorithm running in~$2^{o(|V_U|)}\cdot n^\Oh{1} \subseteq 2^{o(\xi)} \subseteq 2^{o(\xi+\mu)}$ time for \probName{$2$-Balanced~$3$-SAT} which contradicts \Cref{thm:ETH:SAT}.
\end{proof}

In the remainder of this section, we present complexity results concerning various structural restrictions of the topology. Arguably, the most prominent structural parameter is the tree-width of a graph that, informally speaking, expresses its tree-likeness and is usually small in real-life networks~\cite{ManiuSJ2019}. Unfortunately, we can show an intractability result, which already rules out a fixed-parameter tractable algorithm for graphs with a small vertex cover number---an even more restrictive structural parameter than the treewidth.

\begin{theorem}\label{thm:ARH:Wh:vc}
	The \ARH problem is \Wh when parameterized by the vertex cover number~$\operatorname{vc}(G)$ of the topology. Moreover, unless ETH fails, there is no algorithm that solves \ARH in~$f(|\operatorname{vc}(G)|)\cdot n^{o(\sqrt{\operatorname{vc}(G)})}$ time for any computable function~$f$.
\end{theorem}
\begin{proof}
	Although not stated formally, it can be seen that the construction used to prove \Cref{thm:ARH:Wh:I} has not only many parameter-many inhabitants, but these agents also form a vertex cover of the topology. Recall from \Cref{fig:ARH:IR:Wh:I} that each vertex- and edge-selection gadget is an independent set, and these gadgets are connected only via vertices~$G_{i,j}$,~$i,j\in[k]$. Moreover, every edge-selection gadget is additionally connected to a vertex~$M_{i,j}$,~$i,j\in[k]$. If we remove all~$G_{i,j}$ and~$M_{i,j}$, we obtain a disjoint union of independent sets. Hence, these vertices are a vertex cover of size~$\Oh{k^2}$. The running time lower bound follows by the same arguments as in \Cref{thm:ARH:Wh:I}.
\end{proof}

It is well-known and easy to see that whenever a graph is of a bounded vertex cover number, it is also of a bounded treewidth. Therefore, the intractability for treewidth follows directly from the previous theorem. In the following result, we strengthen the hardness and provide an even stricter running-time lower bound.

\begin{theorem}\label{thm:ARH:Wh:tw}
	The \ARH problem is \Wh parameterized by the tree-width~$\operatorname{tw}(G)$ of the topology~$G$. Unless ETH fails, there is no algorithm solving \ARH in~$f(\tau)\cdot n^{o(\tau/\log \tau)}$ time, where~$\tau=\operatorname{tw}(G)$, for any computable function~$f$.
\end{theorem}
\begin{proof}
    We reduce from the \probName{Unary Bin Packing} problem. Here, we are given a bin capacity~$B$, a set of items~$A=(a_1,\ldots,a_n)$ and a number of bins~$k$. Our goal is to decide whether there is an assignment~$\beta$ of all items to bins with respect to bin capacity. It is known that, unless ETH fails, \probName{Unary Bin Packing} cannot be solved in~$g(k)\cdot n^{o(k/\log k)}$ time for any function~$g$~\cite{JansenKMS2013}.
	
	Our construction of an equivalent instance~$\mathcal{I}'$ of the \ARHshort problem is as follows. First, for every~$i\in[k]$ we create a \emph{bin vertex}~$B_i$ and make it occupied by an inhabitant~$b_i$ with~${A_{b_i} = [B]_0}$. Note that these inhabitants approve the whole interval from~$0$ to~$B$.
	Next, we add \emph{items gadgets} to map the items to bins. The single item gadget for some~$a_j\in A$ is a star with~$a_j$ leaves and a center occupied by an inhabitant~$c_j$ with~$A_{c_j} = \{0,a_j\}$. This ensures that refugees occupy all or none of the leaves in the solution. We add an item gadget~$X_j^i$ for every~$i\in[k]$ and every~$a_j\in A$ and add an edge connecting bin vertex~$B_i$ with all leaves of item gadgets~$X_j^i$, where~$i\in[k]$ and~$j\in[n]$. 
	Finally, we must ensure that every item is assigned to exactly one bag. This is ensured by a \emph{guard vertex}~$G_i$ for every item~$a_i\in A$. This vertex is occupied by an inhabitant~$g_i$ with~$A_{g_i} = \{a_i\}$ and is connected to all leaves of item gadgets~$X_i^j$, where~$j\in[k]$. The number of refugees in our instance is~$|R| = \sum_{a_i\in A} a_i$.
	For an illustration of our construction, we refer the reader to \Cref{fig:SDRH:IR:Wh:5pvcn}.
	
	\begin{figure}
        \centering
		\begin{tikzpicture}
			\node[white,fill=black,circle,inner sep=1.5pt,label=90:{\small$[0,B]$}] (B1) at (0,0) {$B_1$};
			
			\node[white,fill=black,circle,inner sep=1.5pt,label=90:{\small$\{0,a_1\}$}] (a1B1) at (2,1) {$X_1^1$};
			\node[draw,circle,inner sep=1pt] (a1B1l1) at (2,0.25) {};
			\node[draw,circle,inner sep=1pt] (a1B1l2) at (1.75,0.25) {};
			\node[draw,circle,inner sep=1pt] (a1B1l3) at (2.25,0.25) {};
			\draw (a1B1) edge (a1B1l1) edge (a1B1l2) edge (a1B1l3);
			\draw[bend right=7] (B1) edge (a1B1l1) edge (a1B1l2) edge (a1B1l3);
			
			\node[white,fill=black,circle,inner sep=1.5pt,label=90:{\small$\{0,a_2\}$}] (a2B1) at (3,1) {$X_2^1$};
			\node[draw,circle,inner sep=1pt] (a2B1l1) at (3.125,0.25) {};
			\node[draw,circle,inner sep=1pt] (a2B1l2) at (2.875,0.25) {};
			\draw (a2B1) edge (a2B1l1) edge (a2B1l2);
			\draw[bend right=7] (B1) edge (a2B1l1) edge (a2B1l2);
			
			\node at (4,1) {$\cdots$};
			
			\node[white,fill=black,circle,inner sep=1.5pt,label=90:{\small$\{0,a_n\}$}] (anB1) at (5,1) {$X_n^1$};
			\node[draw,circle,inner sep=1pt] (anB1l1) at (4.625,0.25) {};
			\node[draw,circle,inner sep=1pt] (anB1l2) at (4.875,0.25) {};
			\node[draw,circle,inner sep=1pt] (anB1l3) at (5.125,0.25) {};
			\node[draw,circle,inner sep=1pt] (anB1l4) at (5.375,0.25) {};
			\draw (anB1) edge (anB1l1) edge (anB1l2) edge (anB1l3) edge (anB1l4);
			\draw[bend right=7] (B1) edge (anB1l1) edge (anB1l2) edge (anB1l3) edge (anB1l4);

			\node[white,fill=black,circle,inner sep=1.5pt,label=90:{\small$[0,B]$}] (B2) at (0,-1.5) {$B_2$};
			
			\node[white,fill=black,circle,inner sep=1.5pt] (a1B2) at (2,-0.5) {$X_1^2$};
			\node[draw,circle,inner sep=1pt] (a1B2l1) at (2,-1.25) {};
			\node[draw,circle,inner sep=1pt] (a1B2l2) at (1.75,-1.25) {};
			\node[draw,circle,inner sep=1pt] (a1B2l3) at (2.25,-1.25) {};
			\draw (a1B2) edge (a1B2l1) edge (a1B2l2) edge (a1B2l3);
			\draw[bend right=7] (B2) edge (a1B2l1) edge (a1B2l2) edge (a1B2l3);
			
			\node[white,fill=black,circle,inner sep=1.5pt] (a2B2) at (3,-0.5) {$X_2^2$};
			\node[draw,circle,inner sep=1pt] (a2B2l1) at (3.125,-1.25) {};
			\node[draw,circle,inner sep=1pt] (a2B2l2) at (2.875,-1.25) {};
			\draw (a2B2) edge (a2B2l1) edge (a2B2l2);
			\draw[bend right=7] (B2) edge (a2B2l1) edge (a2B2l2);
			
			\node at (4,-0.5) {$\cdots$};
			
			\node[white,fill=black,circle,inner sep=1.5pt] (anB2) at (5,-0.5) {$X_n^2$};
			\node[draw,circle,inner sep=1pt] (anB2l1) at (4.625,-1.25) {};
			\node[draw,circle,inner sep=1pt] (anB2l2) at (4.875,-1.25) {};
			\node[draw,circle,inner sep=1pt] (anB2l3) at (5.125,-1.25) {};
			\node[draw,circle,inner sep=1pt] (anB2l4) at (5.375,-1.25) {};
			\draw (anB2) edge (anB2l1) edge (anB2l2) edge (anB2l3) edge (anB2l4);
			\draw[bend right=7] (B2) edge (anB2l1) edge (anB2l2) edge (anB2l3) edge (anB2l4);
			
			\node at (0,-2.075) {$\vdots$};

			\node[white,fill=black,circle,inner sep=1.5pt,label=90:{\small$[0,B]$}] (Bm) at (0,-3.25) {$B_k$};
			
			\node[white,fill=black,circle,inner sep=1.5pt] (a1Bm) at (2,-2.25) {$X_1^k$};
			\node[draw,circle,inner sep=1pt] (a1Bml1) at (2,-3) {};
			\node[draw,circle,inner sep=1pt] (a1Bml2) at (1.75,-3) {};
			\node[draw,circle,inner sep=1pt] (a1Bml3) at (2.25,-3) {};
			\draw (a1Bm) edge (a1Bml1) edge (a1Bml2) edge (a1Bml3);
			\draw[bend right=7] (Bm) edge (a1Bml1) edge (a1Bml2) edge (a1Bml3);
			
			\node[white,fill=black,circle,inner sep=1.5pt] (a2Bm) at (3,-2.25) {$X_2^k$};
			\node[draw,circle,inner sep=1pt] (a2Bml1) at (3.125,-3) {};
			\node[draw,circle,inner sep=1pt] (a2Bml2) at (2.875,-3) {};
			\draw (a2Bm) edge (a2Bml1) edge (a2Bml2);
			\draw[bend right=7] (Bm) edge (a2Bml1) edge (a2Bml2);
			
			\node at (4,-2.5) {$\cdots$};
			
			\node[white,fill=black,circle,inner sep=1.5pt] (anBm) at (5,-2.25) {$X_n^k$};
			\node[draw,circle,inner sep=1pt] (anBml1) at (4.625,-3) {};
			\node[draw,circle,inner sep=1pt] (anBml2) at (4.875,-3) {};
			\node[draw,circle,inner sep=1pt] (anBml3) at (5.125,-3) {};
			\node[draw,circle,inner sep=1pt] (anBml4) at (5.375,-3) {};
			\draw (anBm) edge (anBml1) edge (anBml2) edge (anBml3) edge (anBml4);
			\draw[bend right=7] (Bm) edge (anBml1) edge (anBml2) edge (anBml3) edge (anBml4);

			\node[white,fill=black,circle,inner sep=1.5pt,label=-90:{\small$\{a_1\}$}] (g1) at (1.5,-4) {$G_1$};
			\draw[bend left=7] (g1) edge (a1Bml1) edge (a1Bml2) edge (a1Bml3);
			\draw[bend left=7] (g1) edge (1.4,-3.5) edge (1.3,-3.5) edge (1.2,-3.5) edge (1.1,-3.5);
			
			\node[white,fill=black,circle,inner sep=1.5pt,label=-90:{\small$\{a_2\}$}] (g2) at (2.5,-4) {$G_2$};
			\draw[bend left=7] (g2) edge (a2Bml1) edge (a2Bml2);
			\draw[bend left=7] (g2) edge (2.5,-3.5) edge (2.4,-3.5) edge (2.3,-3.5) edge (2.2,-3.5);
			
			\node at (3.5,-4) {$\cdots$};
			
			\node[white,fill=black,circle,inner sep=1.5pt,label=-90:{\small$\{a_n\}$}] (gn) at (4.5,-4) {$G_n$};
			\draw[bend left=7] (gn) edge (4.4,-3.5) edge (4.3,-3.5) edge (4.2,-3.5) edge (4.1,-3.5);
			\draw[bend left=7] (gn) edge (anBml1) edge (anBml2) edge (anBml3) edge (anBml4);
		\end{tikzpicture}
		\caption{Overview of the construction used in the proof of \Cref{thm:ARH:Wh:tw}. Every guard vertex~$G_i$ is connected to the leaves of each element gadget~$X_i^j$, where~$j\in[k]$.}
		\Description{Overview of the construction used in the proof of \Cref{thm:ARH:Wh:tw}. Every guard vertex~$G_i$ is connected to the leaves of each element gadget~$X_i^j$, where~$j\in[k]$.}
		\label{fig:SDRH:IR:Wh:5pvcn}
	\end{figure}
	
	For the correctness, let~$\mathcal{I}$ be a \yesI of the \probName{Unary Bin Packing} problem and~$\beta$ be a solution assignment. For every~$a_i\in A$, we house~$a_i$ refugees to the leaves of the item gadget~$X_i^{\beta(a_i)}$. Since~$\mathcal{I}$ is a \yesI, every bin vertex has at most~$B$ refugees in the neighborhood. Moreover, each item gadget is either empty or full, in the neighborhood of every guard inhabitant~$g_i$ there is exactly~$a_i$ neighboring refugees, and, finally, all refugees are housed. Thus,~$\mathcal{I'}$ is also a \yesI.
	
	In the opposite direction, let~$\mathcal{I}'$ be a \yesI and~$\pl$ be a solution housing. Due to the definition of the approval sets, every item gadget is either full or empty. In our construction, there are~$k$ copies of item gadgets for every item~$a_i\in A$ and guard vertices secure that exactly one copy is full. Moreover, the bin vertices accept at most~$B$ refugees in their neighborhood. We recall that~$|R| = \sum_{a_i\in A} a_i$. Hence, we define the solution assignment~$\beta$ for~$\mathcal{I}$ as~$\beta(a_i) = j$, where~$j\in[k]$ and~$X_i^j$ is full.
	
	It is not hard to see that the construction has tree-width~$\Oh{k}$; if we remove the vertices~$B_1,\ldots,B_k$ from~$G$, we obtain~$n$ disconnected components~$C_1,\ldots,C_n$. Each of these components consists of~$k$ stars and an extra vertex connected to the leaves of these stars. Hence, the treewidth of each such~$C_j$ is exactly two. As the components are independent, they can be processed independently in the tree decomposition, and the treewidth of~$G$ is, therefore, at most~$k + 2$. Now, assume that there is an algorithm~$\mathbb{A}$ for \ARHshort running in~$f(\tau)\cdot n^{o(\tau/\log \tau)}$, where~$\tau = \operatorname{tw}(G)$. Then, given an instance~$\mathcal{I}$ of \probName{Unary Bin Packing}, we can turn it into an equivalent instance~$\mathcal{I}'$ of \ARHshort using the reduction above, solve~$\mathcal{I}'$ using~$\mathbb{A}$, and return the same response for~$\mathcal{I}$. Since the reduction can be done in polynomial time, we overall obtain an algorithm for \probName{Unary Bin Packing} running in~$f(\tau)\cdot n^{o(\tau/\log \tau)} = g(k)\cdot n^{o(k/\log k)}$, which contradicts ETH.
\end{proof}

We complement the hardness result with respect to the vertex cover number by a matching \XP algorithm. Moreover, if all inhabitants in the vertex cover approve intervals, the algorithm becomes fixed-parameter tractable.

\begin{theorem}\label{thm:ARH:XP:vc}\label{thm:ARH:FPT:vc:interval}
	The \ARH problem is in \XP when parameterized by the vertex cover number~$\operatorname{vc}(G)$. If, additionally, all inhabitants in the vertex cover~$M$ have interval approvals, the problem becomes fixed-parameter tractable.
\end{theorem}
\begin{proof}
	Let~$M\subseteq V$ be a minimum size vertex cover of~$G$ and let~$k = |M|$. The algorithm first guesses (by guessing, we mean iteratively trying all possibilities) the number of refugees~$k' \leq \min\{|V_U \cap M|,|\RF|\}$ that are, in a hypothetical solution, housed on modulator vertices. Additionally, for every~$k'$, we guess~$k'$-sized set~$S \subseteq M \cap V_U$ of particular empty modulator vertices used by a solution housing. Note that there are~$\Oh{k}$ candidate values~$k'$ and~$2^\Oh{k}$ possible sets~$S$. Now, we check whether this housing respects the preferences of all the inhabitants occupying vertices outside of~$M$. If at least one of these inhabitants, say~$h$, disapproves of his neighborhood, we reject the guess. Otherwise, we continue with the second phase of the algorithm.

    If all inhabitants outside of~$M$ are satisfied, they will be satisfied even if we extend the housing with an arbitrary empty vertex outside of the modulator, as the vertices outside of~$M$ form an independent set. Therefore, we can remove all inhabitants outside of~$M$ from the instance to obtain a reduced topology~$G'$. In~$G'$, there are at most~$\Oh{k}$ inhabitants. If all remaining inhabitants approve intervals, we can use a slightly modified ILP formulation from \Cref{thm:ARH:IR:FPT:I:intervals} to verify whether the partial housing~$\pl = S$ can be extended. The modification includes the removal of the variable~$x_v$ for every~$v\in (V_U\cap M)\setminus S$ and adding a condition~$x_v = 1$ for every~$v\in S$. If ILP is feasible, we return \yes. Otherwise, we continue with another guess. As was shown, this ILP can be solved in \FPT time. Since the guessing phase of the algorithm can also be performed in \FPT time, we obtain that the overall algorithm is also \FPT. 
    
    If the inhabitants' preferences are not intervals, we cannot directly use the ILP formulation from the previous case. Instead, for every inhabitant~$h\in \IH$, we guess the number of refugees allocated to its neighborhood~$n_h$, and set~$A'_h = \{n_h\}$. For each such guess of~$n_h$'s, we obtain an instance where the approval set of every inhabitant consists of a single number (and therefore forms an interval), and we can finally use the modified ILP from the previous case to verify whether a housing~$\pl$ can be realized. If the ILP is feasible, we return \yes. Otherwise, we continue with another guess. The bottleneck of this approach is the guessing of different values of~$n_h$, as there are~$n^\Oh{|\IH|}$ different combinations possible. This shows that \ARHshort is in \XP when parameterized by the vertex cover number~$\operatorname{vc}(G)$.
\end{proof}

By combining arguments from \Cref{thm:ARH:IR:FPT:I:delta} with the algorithm of \Cref{thm:ARH:FPT:vc:interval}, we obtain the following last positive result of this section.

\begin{corollary}
	The \ARH problem is fixed-parameter tractable when parameterized by the vertex cover number~$\operatorname{vc}(G)$ and the maximum number of disjoint intervals~$\delta$, combined.
\end{corollary}

\section{Hedonic Preferences}\label{sec:hedonic_setting}

Our second model of refugee housing improves upon the previous model by introducing the individual preferences of refugees. Naturally, refugees are no longer anonymous, and the identity of every particular refugee matters. The preferences of the inhabitants are again dichotomous, and for every inhabitant~$h\in\IH$, the approval set~$A_h$ is a subset of~$2^\RF$. Similarly, for a refugee~$r\in \RF$, the approval set~$A_r$ is a subset of~$2^\IH$. Our goal is to find housing compatible with both groups' preferences.

\begin{definition}
	A housing~$\pl\colon\RF\to V_U$ is called \emph{respecting} if for every~$h\in\IH$ we have~$V_\pl \cap N_G(\ia(h)) \in A_h$ and for every~$r\in\RF$ we have~$N_G(\pl(r)) \cap V_\IH \in A_r$.
\end{definition}

In other words, a housing~$\pl$ is respecting if every inhabitant and every refugee approves its neighborhood. We study the problem of deciding whether there is a respecting housing in the instance with hedonic preferences under the name \HRH (\HRHshort for short).

\begin{remark}
    Observe that, under hedonic preferences, our agents' approvals can be exponential in the input size. Therefore, we cannot represent them explicitly, as in the case of anonymous preferences. Instead, we assume that the approval set of each agent is given in the form of an oracle that, for a given subset of agents of the opposite type, returns \texttt{true} if this set is approved and \texttt{false} otherwise. With this assumption, we measure the running time of our algorithms in the number of oracle calls.
\end{remark}

We start with a simple example that illustrates the definition of hedonic preferences and respecting housing.

\begin{example}\label{ex:HRH}
    \begin{figure}
        \centering
        \begin{tikzpicture}
            \node[draw,white,fill=black,circle,minimum width=0.7cm,label=180:{\normalsize$\{\{r_1\},\underline{\{r_2\}}\}$}] (v1) at (0,0) {$h_1$};
            \node[draw,white,fill=black,circle,minimum width=0.7cm,label=180:{\normalsize$\{\{r_2\},\{r_1,r_2\}\}$}] (v2) at (0,1.5) {$h_2$};
            \node[draw,circle,minimum width=0.7cm,label=0:{\normalsize$\{\{h_1\}\}$}] (v3) at (1.5,1.5) {$r_1$};
            \node[draw,circle,minimum width=0.7cm,label=0:{\normalsize$\{\{h_2\}\}$}] (v4) at (1.5,0) {$r_2$};

            \node at (0.75,-1) {\textcolor{red!80!black}{\Large\xmark}};
            
            \draw (v1) -- (v2) -- (v3) -- (v4) -- (v1);
        \end{tikzpicture}
        \hspace{1cm}
        \begin{tikzpicture}
            \node[draw,white,fill=black,circle,minimum width=0.7cm,label=180:{\normalsize$\{\underline{\{r_1\}},\{r_2\}\}$}] (v1) at (0,0) {$h_1$};
            \node[draw,white,fill=black,circle,minimum width=0.7cm,label=180:{\normalsize$\{\underline{\{r_2\}},\{r_1,r_2\}\}$}] (v2) at (0,1.5) {$h_2$};
            \node[draw,circle,minimum width=0.7cm,label=0:{\normalsize$\{\underline{\{h_2\}}\}$}] (v3) at (1.5,1.5) {$r_2$};
            \node[draw,circle,minimum width=0.7cm,label=0:{\normalsize$\{\underline{\{h_1\}}\}$}] (v4) at (1.5,0) {$r_1$};
            
            \node at (0.75,-1) {\textcolor{green!80!black}{\Large\cmark}};
            
            \draw (v1) -- (v2) -- (v3) -- (v4) -- (v1);
        \end{tikzpicture}
        \caption{An instance of the \HRHshort problem from \Cref{ex:HRH} and two possible housings. On the left, we have housing that is not respecting, as, e.g., refugee $r_1$ does not approve inhabitant $h_2$ in its neighborhood. On the right, the housing is respecting as each agent approves its neighborhood. Underlined elements of the approval sets represent agents' neighborhoods in the respective housing.}
        \label{fig:HRH:example}
    \end{figure}
	Let the topology be a cycle with four vertices. There are two inhabitants~$h_1$ and~$h_2$ assigned to neighboring vertices and two refugees~$r_1$ and~$r_2$ to house. The approval set of inhabitant~$h_1$ is~$A_{h_1} = \{\{r_1\},\{r_2\}\}$. That is,~$h_1$ approves only one refugee in her neighborhood regardless of the identity. The second inhabitant approves set~$A_{h_2}=\{\{r_2,\},\{r_1,r_2\}\}$. In other words, the inhabitant~$h_2$ is dissatisfied with having only the refugee~$r_1$ in the neighborhood and is fine with neighboring both the refugees or $r_2$ alone. For the refugees, we have~$A_{r_1} = \{\{h_1\}\}$ and~$A_{r_2}=\{\{h_2\}\}$. See \Cref{fig:HRH:example} for two possible housings.
\end{example}

Observe that since both inhabitants and refugees have preferences only over the other set of individuals, we can remove all edges between two empty or two occupied vertices, respectively. Hence, all graphs assumed in this section are again bipartite.

Our first result settles a relation between the \ARHshort problem and the setting with hedonic preferences. Not surprisingly, hedonic preferences are more general and can capture an arbitrary instance with anonymous preferences with the same topology, sets of agents, and allocation of inhabitants. However, this general reduction can lead to preferences of size exponential in the input size.

\begin{theorem}\label{thm:HRH:NPh}\label{thm:HRH:ARH:relation}
	Every instance~$\mathcal{I} = (G,\RF,\IH,\ia,(A_h)_{h\in \IH})$ of the \ARH problem is polynomial-time reducible to an equivalent instance~$\mathcal{J} = (G,\RF,(A'_r)_{r\in\RF},\IH,\ia,(A'_h)_{h\in\IH})$ of the \HRH problem.
\end{theorem}
\begin{proof}
	Let~$\mathcal{I}$ be an instance of \ARHshort. We construct an equivalent instance~$\mathcal{J}$ of \HRHshort with the same topology~$G$, set of refugees~$\RF$, set of inhabitants~$\IH$, and the same allocation~$\ia$ of inhabitants as follows. In~$\mathcal{J}$, the approval set~$A'_r$ for every refugee~$r\in\RF$ is simply~$2^\IH$. That is, the refugees approve arbitrary neighborhoods, and the corresponding oracle simply returns \texttt{true} for an arbitrary subset of~$\IH$. Next, let~$h\in\IH$ be an inhabitant and~$A_h$ its approval set in~$\mathcal{I}$. We set~$A'_h = \{ X\in 2^{\RF} \mid |X| \in A_h \}$. In an oracle, such preferences can be simply represented as a list of~$A_h$'s for all inhabitants. This finishes the construction.

    For correctness, let~$\mathcal{I}$ be a \yesI and~$\pl$ be an inhabitant-respecting housing. We show that~$\pl$ is also a solution for~$\mathcal{J}$. As~$A'_r = 2^{\IH}$ for every~$r\in\RF$,~$\pl$ clearly respects all refugees. Let~$h\in\IH$ be an inhabitant and assume that~$\pl^{-1}(N_G(\ia(i))) \not\in A'_h$. We defined~$A'_h$ as a set~$\{ X\in 2^{\IH} \mid |X| \in A_h \}$. Therefore,~$\pl^{-1}(N_G(\ia(i))) \not\in A'_h$ implies that~$|\pl^{-1}(N_G(\ia(i)))| \not\in A_h$, which contradicts that~$\pl$ is a solution for~$\mathcal{I}$. This is not possible, so~$\pl$ is also a solution for~$\mathcal{J}$. In the opposite direction, let~$\mathcal{J}$ be a \yesI and~$\pl'$ be a solution. We again show that~$\pl'$ is also a solution for~$\mathcal{I}$. In \ARHshort, refugees have no preferences, so we must only verify that~$|N_G(\ia(h))\cap \pl| \in A_h$ for every~$h\in \IH$. However, this is trivially satisfied from the definition of~$A'_h$. Thus,~$\pl'$ is a solution for~$\mathcal{I}$, and the reduction is correct.
\end{proof}

By the previous theorem, any hardness result proved in the previous section directly carries over to the setting of \HRHshort. This raises two questions. First, can we strengthen the intractability for certain parameters? Second, which of our tractability results can be generalized to the more general model of preferences? We start with the former question and show a very strong hardness result for \HRHshort.

\begin{theorem}\label{thm:HRH:NPh:vc}\label{thm:HRH:NPh:I}
    The \HRH problem is \NPc even if~$\operatorname{vc}(G) = |\IH| = 1$.
\end{theorem}
\begin{proof}
    We again reduce from \probName{$2$-Balanced~$3$-SAT}. Formally, given an instance~$\mathcal{I}$ of \probName{$2$-Balanced~$3$-SAT}, we create an equivalent instance~$\mathcal{J}$ of \HRHshort as follows. The topology~$G$ is a disjoint union of a star with center~$c$ and~$\xi$ leaves~$\ell_1,\ldots,\ell_\xi$ and~$\xi$ isolated vertices~$v_1,\ldots,v_\xi$. There is a single inhabitant~$h$ assigned to the center~$c$ of the star. The set~$\RF$ contains one refugee~$r_i$ with~$A_{r_i} = \{\emptyset,\{h\}\}$ for every variable~$x_i$,~$i\in[\xi]$. Finally, the inhabitant~$h$ approves a set~$R' \subseteq \RF$ if and only if~$\varphi(\alpha) = 1$, where~$\alpha(x_i) = 1$ if and only if~$r_i\in R'$. Note that we do not need to precompute all the preferences beforehand; instead, the oracle can decide whether a set of neighbors is approved or not on the fly based on the set of refugees provided and the formula $\varphi$.

    For correctness, let~$\mathcal{I}$ be a \yesI and~$\alpha$ be a solution. We construct a housing~$\pl$ such that for every~$r_i$ such that~$\alpha(x_i) = 1$, we set~$\pl(r_i) = \ell_i$ and~$\pl(r_o) = v_i$ otherwise. The refugees assigned to isolated vertices are clearly satisfied as~$\emptyset \in A_{r_i}$ for every~$r\in\RF$. The remaining refugees are neighbors only with the inhabitant~$h$, and~$\{h\}$ is in their approval sets. It remains to verify that~$\pl$ is respecting~$h$'s preferences. Let~$R'$ be a set of refugees housed on the leaves of the star, and, for the sake of contradiction, assume that~$R'\not\in A_h$. It means that~$\phi(\alpha) = 0$, which contradicts that~$\alpha$ is a solution for~$\mathcal{I}$. That is, it must be the case that~$R'\in A_h$. Hence,~$\pl$ is a solution for~$\mathcal{J}$. In the opposite direction, let~$\pl$ be a respecting housing in~$\mathcal{J}$. For every variable~$x_i$, we set~$\alpha(x_i) = 1$ if~$\pl(r_i) \in \{\ell_1,\ldots,\ell_\xi\}$ and~$\alpha(x_i) = 0$ otherwise. If it holds that~$\phi(\alpha) = 0$, we have a direct contradiction with~$\pl$ being a respecting housing for~$h$. That is, it must holds that~$\phi(\alpha) = 1$, meaning that~$\alpha$ is a solution for~$\mathcal{I}$.
\end{proof}

Maybe surprisingly, the isolated vertices in the previous result are necessary for the instance to be intractable, as we show that if the topology~$G$ is a simple star, then \HRHshort is tractable.

\begin{proposition}
    If the topology~$G$ is a star, the \HRH problem can be decided in polynomial time.
\end{proposition}
\begin{proof}
    Let~$c$ be the center of~$G$. Recall that we can assume that $G$ is a bipartite graph with one part consisting of empty vertices and the other part consisting of occupied vertices. We distinguish two cases:~$c$ is empty, or an inhabitant~$h$ exists such that~$\ia(h) = c$. In the latter case, if~$\RF\not\in A_h$ or there exists a refugee~$r\in\RF$ such that~$\{h\}\not\in A_r$, then we are dealing with a clear \noI. Otherwise, we return \yes. In the former case, the center is the only empty vertex. Therefore, we must have~$|\RF| = 1$. The only possible housing is~$\pl(r) = c$, and it can be checked whether it is respecting in polynomial time.
\end{proof}

In our next result, we show that an \XP algorithm exists if the number of refugees to house is constant. Moreover, if we instead parameterize by the number of empty houses~$|V_U|$, the same algorithm becomes fixed-parameter tractable.

\begin{theorem}\label{thm:HRH:XP:R}
	The \HRH problem is in \XP when parameterized by the number of refugees~$|\RF|$ and fixed-parameter tractable when parameterized by~$|V_U|$. Unless ETH fails, there is no algorithm for \HRH running in $f(|\RF|)\cdot n^\oh{|\RF|}$ for any computable function $f$.
\end{theorem}
\begin{proof}
	The algorithm is a simple brute-force. First, we guess for each refugee~$r\in \RF$ its vertex~$v_r$ on which the refugee is housed in a hypothetical solution. Then, we verify whether all guessed vertices are distinct. If not, we reject the guess. Otherwise, we check if all refugees and all inhabitants approve their neighborhoods. If this is the case, we have a solution, and we return \yes. Otherwise, we continue with another guess. If no possibility leads to the \yes response, we return \no. The algorithm tries all possible housings and, therefore, is clearly correct. There are~$|V_U|^{|\RF|}$ possible housings, and the validity of each housing can be checked in polynomial time. Consequently, the running time of the algorithm is~$n^\Oh{|\RF|}$, which is clearly in \XP. If we parameterize by the number of empty vertices~$|V_U|$, then it holds that~$|\RF|\leq|V_U|$. This implies the running time of~$|V_U|^\Oh{|V_U|}\cdot n^\Oh{1}$, which is in \FPT. The lower bound for running time follows from \Cref{thm:HRH:ARH:relation}, and the same lower bound showed for \ARHshort (cf. \Cref{thm:ARH:XP:R}).
\end{proof}

We conclude with an efficient algorithm for the number of refugees and the number of inhabitants, combined.

\begin{theorem}\label{thm:HRH:FPT:R+I}
	The \HRH problem is in \FPT when parameterized by the number of refugees~$|\RF|$ and the number of inhabitants~$|\IH|$, combined.
\end{theorem}
\begin{proof}
	First, we observe that we can partition the empty vertices into at most~$2^{|\IH|}$ types based on their neighborhood. For both refugees and inhabitants, two empty vertices of the same type are clearly indistinguishable since, for a solution, they care only about their neighborhoods and not particular locations. Based on this observation, we can guess for each refugee~$r\in\RF$ a type of vertex on which she is housed in a hypothetical solution~$\pl$. Once we know the type of vertex for each refugee, we can verify that such a housing~$\pl$ is respecting. If this is the case, we return \yes. Otherwise, we continue with another guess. If no housing~$\pl$ is respecting, the algorithm returns \no. There are~$(2^\Oh{|\IH|})^{|\RF|}$ possible guesses, and each one can be verified in polynomial time. Hence, the theorem follows.
\end{proof}

\section{Diversity Preferences}\label{sec:diversity_setting}

In the anonymous refugee housing, we are not assuming the preferences of individual refugees. Thanks to this property, the model is as simple as possible. The fully hedonic setting from \Cref{sec:hedonic_setting} precisely captures the preferences of both the refugees and the inhabitants. On the other hand, the fully hedonic model is not very realistic, as it is hard to acquaint all inhabitants with all refugees.

Hence, we introduce the third model of refugees housing, where both the inhabitants and the refugees are partitioned into types and agents from both groups have preferences over fractions of agents of each type in their neighborhood. 

Such diversity goals, where agents are partitioned into types and the preferences of agents are based on the fraction of each type in their neighborhood or coalition, were successfully used in many scenarios such as school choice~\cite{AzizZ2021,AzizGSW2019,AzizGS2020}, public housing~\cite{BenabbouCHSZ2018,GrossHumberBBM2021}, stable rommate~\cite{BoehmerE2020a}, hedonic games~\cite{BredereckEI2019,BoehmerE2020b,Darmann2023,GanianHKSS2023}, multi-attribute matching~\cite{AhmadiADFK2021}, or employee hiring~\cite{SchumannCFD2019}.

Before we formally define the computational problem of our interest, let us introduce further notation. Let~$\agents = \IH\cup\RF$ be a set of agents partitioned into~$\tau$ types~$T_1,\ldots,T_\tau$. For a set~$S\subseteq \agents$, we define a \emph{palette} as a~$\tau$-tuple~$\left(\frac{|T_i\cap S|}{|S|}\right)_{i\in[\tau]}$ if~$|S|\geq 1$ and~$\tau$-tuple~$(0,\ldots,0)$ if~$S=\emptyset$. Given an agent~$a\in \agents$, her approval set is a subset of the set~$\left\{\left(\frac{|T_i\cap S|}{|S|}\right)_{i\in[\tau]}\mid S\subseteq 2^\agents\right\}$.

\begin{definition}
	A housing~$\pl\colon\RF\to V_U$ is called \emph{diversity respecting} if for every inhabitant~$i\in\IH$ the palette for the set~$(V_{\IH} \cap N_G(\ia(i)))\cup( V_{\pl} \cap N_G(\ia(i)) )$ is in~$A_i$, and for every refugee~$r\in\RF$ the palette for the set~$(V_{\IH}\cap N_G(\pl(r)))\cup(V_{\pl} \cap N_G(\pl(r)) )$ is in~$A_r$.
\end{definition}

The \DRH problem (\DRHshort for short) then asks whether there is a diversity respecting housing~$\pl$. Note that this time, we are not allowed to drop edges between two inhabitants or two empty houses, and thus, the graphs assumed in this section are no longer bipartite.

\begin{example}\label{ex:DRH}
    \begin{figure}
        \centering
        \begin{tikzpicture}
            \node[draw,white,fill=black,circle,minimum width=0.7cm,label=180:{\normalsize$\{(1,0),(1/2,1/2)\}$}] (v1) at (0,0) {$h_1$};
            \node[draw,white,fill=black,circle,minimum width=0.7cm,label=180:{\normalsize$\{(1,0)\}$}] (v2) at (0,1.5) {$h_2$};
            \node[draw,circle,minimum width=0.7cm,label=0:{\normalsize$\{(1,0)\}$}] (v3) at (1.5,1.5) {$r$};
            \node[draw,circle,minimum width=0.7cm] (v4) at (1.5,0) {};

            \node at (0.75,-1) {\textcolor{red!80!black}{\Large\xmark}};
            
            \draw (v1) -- (v2) -- (v3) -- (v4) -- (v1);
        \end{tikzpicture}
        \hspace{0.5cm}
        \begin{tikzpicture}
            \node[draw,white,fill=black,circle,minimum width=0.7cm,label=180:{\normalsize$\{(1,0),(1/2,1/2)\}$}] (v1) at (0,0) {$h_1$};
            \node[draw,white,fill=black,circle,minimum width=0.7cm,label=180:{\normalsize$\{(1,0)\}$}] (v2) at (0,1.5) {$h_2$};
            \node[draw,circle,minimum width=0.7cm] (v3) at (1.5,1.5) {};
            \node[draw,circle,minimum width=0.7cm,label=0:{\normalsize$\{(1,0)\}$}] (v4) at (1.5,0) {$r$};
            
            \node at (0.75,-1) {\textcolor{green!80!black}{\Large\cmark}};
            
            \draw (v1) -- (v2) -- (v3) -- (v4) -- (v1);
        \end{tikzpicture}
        \caption{An instance of the \DRHshort problem from \Cref{ex:DRH} and two possible housings. On the left, we have housing that is not respecting. While inhabitant $h_2$ is satisfied as in its neighborhood, there are only agents of type $T_1$. This is not true for the remaining agents. They have only agent $h_2$ in the neighborhood, and the palette corresponding to $\{h_2\}$ is $(0,1)$, which is in the approval set of neither $h_1$ nor $r$. The housing on the right is, on the other hand, respecting.}
        \label{fig:DRH:example}
    \end{figure}
	Let the topology be a cycle with four vertices. There are two agents of type~$T_1$. One of these agents is an inhabitant~$h_1$ approving~$\{(1,0),(1/2,1/2)\}$ and the second one is a refugee~$r$ approving only agents of his own type; that is,~$A_r = \{(1,0)\}$. The type~$T_2$ contains one inhabitant~$h_2$ approving the set~$\{(1,0)\}$. Inhabitants are assigned such that they are neighbors. There are two possible housings for the refugee~$r$. She can be either a neighbor of~$h_1$ or~$h_2$. Since she accepts only agents of her own type in the neighborhood, the only diversity-respecting housing is next to inhabitant~$h_1$. See \Cref{fig:DRH:example} for a more detailed illustration of possible housings.
\end{example}

We start the investigation of the computational complexity of the \DRHshort problem by establishing a general connection between the hedonic and diversity preferences. More specifically, we show that \DRHshort is at least as hard as the \HRHshort problem. Therefore, every intractability result from the previous section is directly carried over.

\begin{theorem}\label{thm:DRH:HRH:reduction}
    Every instance~$\mathcal{I} = (G,\RF,(A_r)_{r\in\RF},\IH,\ia,(A_h)_{h\in \IH})$ of the \HRH problem is polynomial-time reducible to an equivalent instance~$\mathcal{J} = (G,\RF,(A'_r)_{r\in\RF},\IH,\ia,(A'_h)_{h\in\IH},(T_i)_{i\in\tau})$ of the \DRH problem with~$\tau = |\RF| + |\IH|$.
\end{theorem}
\begin{proof}
    Let~$\mathcal{I}$ be an instance of the \HRHshort problem. Without loss of generality, we can assume that the topology~$G$ is bipartite, with one part containing only occupied vertices and the other part containing only empty vertices. We construct an equivalent instance~$\mathcal{J}$ of \DRHshort as follows. We fix a bijection~$\beta\colon \RF\cup\IH \to [|\RF| + |\IH|]$ and set for each agent~$a\in \RF\cup\IH$ its type as~$T_{\beta(a)}$. Observe that since~$\beta$ is a bijection, each type consists of exactly one agent. Now, we define the preferences. Let~$r\in\RF$ be a refugee and~$A_r$ be its approval set. For every~$X\in A_r$, we add to~$A'_r$ the tuple~$(f_1,\ldots,f_\tau)$, where~$f_i=1/|X|$ if~$\beta^{-1}(i) \in X$ and~$0$ otherwise. For every inhabitant~$h\in\IH$, we define the approval set~$A'_h$ analogously.

    For correctness, assume that~$\mathcal{I}$ is a \yesI and~$\pl$ is a respecting housing. We claim that~$\pl$ is also a solution for~$\mathcal{J}$. For the sake of contradiction, assume that there is a refugee~$r$ such that the palette for~$\{h\in\IH\mid \ia(h) \in N_G(\pl(r))\}\cup\{r'\in\RF\mid \pl(r')\in N_G(\pl(r))\}$ is not in~$A'_r$. First, we can observe that the set~$\{r'\in\RF\mid \pl(r')\in N_G(\pl(r))\}$ is always empty, as the topology~$G$ contains no edge between two empty vertices. As each agent is of a different type, the palette for~$r$'s neighborhood is of the form~$(f_1,\ldots,f_t)$, where~$f_i = 1/|X|$ whenever~$\beta^{-1}(i) \in \{h\in\IH\mid \ia(h) \in N_G(\pl(r))\} = X$ and~$0$ otherwise. However, by the construction of~$A'_r$, such a palette is not in~$A'_r$ if and only if~$X\notin A_r$. This contradicts that~$\pl$ is a solution for~$\mathcal{I}$. For inhabitants, the proof is again analogous. Thus,~$\pl$ is indeed a solution for~$\mathcal{J}$. In the opposite direction, let~$\mathcal{J}$ be a \yesI and~$\pl'$ be a solution. Again, we claim that~$\pl$ is also a solution for~$\mathcal{I}$. For the sake of contradiction, let there be a refugee~$r\in\RF$ (analogously for an inhabitant~$h\in\IH$) for whom~$\pl$ is not respecting in~$\mathcal{I}$. Let~$X\subseteq\IH$ be a subset of inhabitants in the neighborhood of~$\pl(r)$. Since~$X\not\in A_r$, also~$(f_1,\ldots,f_\tau)$, where~$f_i = 1/|X|$ if~$\beta^{-1}(i) \in X$ and~$0$ otherwise, is not in~$A'_r$. However, this contradicts the claim that~$\pl$ is a solution for~$\mathcal{J}$. Hence,~$\mathcal{I}$ is indeed a \yesI, which completes the proof.
\end{proof}

In the general reduction from \Cref{thm:DRH:HRH:reduction}, we heavily exploit the number of types to ensure that the constructed instance of \DRHshort is indeed equivalent to the input instance of \HRHshort. One can argue that, for real-life instances, it is reasonable to expect that the number of different types is very small compared to the number of inhabitants and refugees. However, as we show in the following result, even if there are only two types of agents, the \DRHshort problem remains computationally intractable.

\begin{theorem}\label{thm:DRH:NPc:types}
	The \DRH problem is \NPc, even if there are only two types of agents, the topology is a bipartite graph, and every agent approves exactly one palette.
\end{theorem}
\begin{proof}
	We show the \NPhness by a reduction from the \probName{Set Cover} problem, which is known to be \NPc~\cite{GareyJ1979}. In this problem, we are given a universe~$U = \{u_1,\ldots,u_n\}$, a family~$\mathcal{F}$ of subsets of~$U$, and an integer~$k\in\N$. The goal is to decide whether there is a sub-family~$\mathcal{C}\subseteq\mathcal{F}$ of size at most~$k$ such that~$\bigcup_{C\in\mathcal{C}} C = U$.
	
	Given an instance~$\mathcal{I}=(U,\mathcal{F},k)$, we construct an equivalent instance~$\mathcal{I'}$ of \DRHshort as follows. For every element~$u_i\in U$, we add one vertex~$v_i$ and assign to it an inhabitant~$h_i$. The inhabitant~$h_i$ is of type~$T_1$ and his approval set is~$\{(0,1)\}$. Next, for every subset~$F\in\mathcal{F}$, we create one vertex~$v_F$ that is adjacent to every~$v_i$ such that~$u_i\in F$. To finalize the construction, we add~$k$ refugees~$r_1,\ldots,r_k$ of type~$T_2$ approving the set~$\{(1,0)\}$.
	
	For the correctness, let~$\mathcal{I}$ be a \yesI and~$\mathcal{C} = \{C_1,\ldots,C_k\}$ be a solution for~$\mathcal{I}$. For every~$i\in[k]$, we house the refugee~$r_i$ on the vertex corresponding to the set~$C_i$. In this housing, every refugee is satisfied and, since~$\mathcal{C}$ is a set cover, every~$h_i\in \IH$ neighbors with at least one refugee. In the opposite direction, let~$\pl$ be a diversity-respecting housing. We add to~$\mathcal{C}$ a set~$C_i\in\mathcal{C}$ if and only if there is a refugee housed on the corresponding vertex~$v_{C_i}$. Suppose that there is an element~$u\in U$ which is not covered by~$\mathcal{C}$. Then, there is an inhabitant who is not a neighbor of any refugee. However, this could not be the case as~$\pl$ is diversity respecting. Hence, the reduction is correct and can clearly be done in polynomial time.
\end{proof}

Next, we show a tractable algorithm for the setting with diversity preferences and small number of refugees. Interestingly, even though the algorithm is a simple brute-force over all possible housings, it is, under the standard theoretical assumptions, asymptotically optimal.

\begin{theorem}\label{thm:DRH:XP:R}
	The \DRH problem is in \XP when parameterized by the number of refugees~$|\RF|$ and is fixed-parameter tractable when parameterized by the number of empty vertices~$|V_U|$. Unless ETH fails, there is no algorithm running in~$f(|\RF|)\cdot n^\oh{|\RF|}$ time for any computable function~$f$.
\end{theorem}
\begin{proof}
    A housing is a mapping between the set of refugees~$\RF$ and the set of empty vertices~$V_U$. Therefore, we can brute-force over all~$|V_U|^\Oh{|\RF|}$ possible housings~$\pl$, and for each of them, verify in polynomial time that~$\pl$ is indeed respecting. If we encounter at least one respecting housing, we return \yes. Otherwise, if no housing is respecting, we return \no. The algorithm is trivially correct as it tries all possible solutions and its running time is~$|V_U|^\Oh{R} \cdot n^\Oh{1}$, which is in \XP with respect to~$|\RF|$, and in \FPT with respect to~$|V_U|$, as~$|\RF| \leq |V_U|$. The running time lower-bound follows from \Cref{thm:DRH:HRH:reduction} and analogous result for \HRHshort.
\end{proof}

Naturally, the previous result yields an \XP algorithm for the parameterization by the combined parameter the number of inhabitants and the number of refugees. We conclude with one more intractability result that shows that, unlike in the two previous models, in the case of \DRHshort, an \FPT algorithm for these two parameters is unlikely.

\begin{theorem}\label{thm:DRH:Wh:IR}
    The \DRH problem is \Wh when parameterized by the number of inhabitants~$|\IH|$ and the number of refugees~$|\RF|$, combined, even if each agent approves only one palette.
\end{theorem}
\begin{proof}
    We show the hardness by a reduction from the \probName{Multicolored Clique}. Recall that in this problem, we are given a~$k$-partite graph~$H=(V_1,\ldots,V_k,E)$, the goal is to find a clique~$K$ of size~$k$, and the problem is \Wc with respect to~$k$~\cite{FellowsHRV2009}. Given an instance~$\mathcal{I} = (H,k)$ of \probName{Multicolored Clique}, we construct an equivalent instance~$\mathcal{J}$ of \DRHshort as follows.

    First, we define the topology~$G$. We start with the graph~$H$ and, for each part~$V_i$, we add one pendant vertex~$g_i$ and make it occupied by an inhabitant~$h_i$. The set~$\RF$ contains one refugee~$r_i$ for each color~$i\in[k]$. There are~$\tau=k$ types~$T_1,\ldots,T_\tau$ so that~$T_i = \{h_i,r_i\}$. Finally, we need to define the preferences of our agents. Every inhabitant~$h_i$ approves only the palette~$(f_1,\ldots,f_\tau)$, where~$f_j = 1$ if~$j=i$ and~$0$ otherwise. In other words, the inhabitants require only the refugee of their own color to be in their neighborhood. For a refugee~$r_i$, we set~$A_r = (1/k,\ldots,1/k)$; that is, each refugee requires the same (non-zero) number of agents of each type in its neighborhood. This finishes the description of the construction.

    For correctness, let~$\mathcal{I}$ be a \yesI and~$K$ be a clique of size~$k$ in~$H$. We construct a housing~$\pl$ so that for every~$r_i$, we set~$\pl(r_i) = K\cap V_i$. Clearly,~$\pl$ is respecting for every inhabitant, as they are neighbors with exactly the refugee of their own type. If there is a refugee~$r$ for which~$\pl$ is not respecting, then there is no agent of some type~$t\in[\tau]$ in~$r$'s neighborhood. By the construction of~$\pl$,~$r$ is clearly not missing a neighbor of its own type, so it must be an agent of a different type. This, however, contradicts that~$K$ is a clique in~$H$. Therefore,~$\pl$ is respecting and~$\mathcal{J}$ is a \yesI.

    In the opposite direction, let~$\mathcal{J}$ be a \yesI and~$\pl$ be a solution housing. We set~$K = \{ v \in V(H) \mid \exists r\in\RF\colon \pl(r) = v\}$. First, observe that~$|K| = k$, as there is no empty vertex in~$V(G)\setminus V(H)$. Now, assume that~$K$ does not contain a vertex of some part~$V_i$. Then, by the Pigeonhole principle, there is at least one part~$V_j$ such that~$|V_j \cap K| \geq 2$. Since the inhabitant~$h_j$ requires only agents of its own type in its neighborhood, all refugees housed on~$V_j$ must be of the same type. This is, however, not possible, as there is only one refugee of each type. That is,~$\pl$ is not respecting for~$h_j$, which is a contradiction. Consequently,~$K$ contains exactly one vertex from each part. Finally, assume that~$K$ is not a clique; that is, there is a pair of vertices~$u,v\in K$ such that~$\{u,v\}\not\in E(H)$. Let~$r_i$ be the refugee that caused~$u\in K$ and~$r_j$ be the refugee that caused~$v\in K$. Since~$\{u,v\}\not\in E(G)$, the refugee~$r_i$ does not have a neighbor of type~$T_j$, which again contradicts that~$T_j$ is respecting. Therefore,~$\{u,v\}$ must be an edge of~$G$.

    To conclude, observe that the reduction can be clearly done in polynomial time and that~$|\IH| + |\RF| = 2k$. That is, the reduction is indeed a parameterized reduction, finishing the proof.
\end{proof}

\section{Conclusions} 

We initiated the study of a novel model of refugee housing. The model mainly targets the situations where refugees need to be accommodated and integrated in the local community. This distinguishes us from the previous settings of refugee resettlement, where the perspective is more global.

Our results identify some tractable and intractable cases of finding stable outcomes from the viewpoint of both the classical computational complexity and the finer-grained framework of parameterized complexity. To this end, we believe that other notions of stability inspired, for example, by the model of Schelling games of \citet{AgarwalEGISV2021} or by exchange-stability of~\citet{Alcalde1994}, should be investigated. The first steps in this direction were done by \citet{LisowskiS2023}, who studied swap-stability in the context of refugee housing.

A natural way to tackle the intractability of problems in computational social choice is to restrict the preferences of agents~\cite{ElkindLP2022}. One such restriction that should be investigated, especially in the case of an anonymous setting, is the voter-interval and candidate interval, which were successfully used in similar scenarios; see, e.g., \cite{BiloBLM2022a,Yang2017,BredereckEI2019} or the survey of \citet{ElkindLP2022}. Besides that, we are interested in the anonymous setting in which every inhabitant~$i\in\IH$ approves an interval~$[0,u_i]$, where~$u_i \geq 0$ is an inhabitant-specific upper-bound on the number of refugees in her neighborhood. This direction was already investigated in the follow-up paper by \citet{Schierreich2023}.

Next, there are many notions measuring the quality of an outcome studied in the literature in both the context of Schelling and hedonic games~\cite{ElkindFF2020,AgarwalEGISV2021,AzizCGS2018}, and we believe that these notions should be investigated even in the context of refugee housing. In this line, the most prominent notion is the utilitarian social welfare of an outcome.

Finally, our approach is mostly theoretical. For a practical applicability of the model, many improvements of the model need to be done. The necessary steps are extensively discussed in the work of Schierreich~\cite{Schierreich2024}.

\bibliographystyle{ACM-Reference-Format}
\bibliography{references}

\end{document}